%% file: IsogeometricVolumeVisualization_Fuchs_Hjelmervik.tex
\documentclass[compsoc]{IEEEtran}

\usepackage{caption}

\usepackage{hyphenat}
\usepackage{amssymb}
\usepackage{amsmath}
\usepackage{latexsym}
\usepackage[hidelinks]{hyperref}
\usepackage[hang,nooneline]{subfigure}
\usepackage{dsfont}
\usepackage[table]{xcolor}
\usepackage{algorithm}
\usepackage{algorithmicx}
\usepackage{algpseudocode}
\usepackage{mathtools}
\usepackage{wrapfig}

\newcommand{\reffig}[1]{\figurename~\ref{fig:#1}}
\newcommand{\refequ}[1]{Equation~\eqref{eq:#1}}
\newcommand{\refsec}[1]{Section~\ref{sec:#1}}
\newcommand{\refdef}[1]{Definition~\eqref{def:#1}}
\newcommand{\refalg}[1]{Algorithm~\ref{alg:#1}}
\newcommand{\reftable}[1]{\tablename~\ref{table:#1}}

\usepackage{xspace}
\newcommand{\Bezier}{B\'ezier\xspace}

\usepackage{amsthm}
\newtheorem{theorem}{Theorem}
\newtheorem{definition}{Definition}

\newcommand{\R}{\ensuremath{\mathbb{R}}}

\newcommand{\glow}{\sigma}

\newcommand{\gfront} {g_{In}}
\newcommand{\tildegfront} {{\tilde g}_{In}}
\newcommand{\gback} {g_{Out}}
\newcommand{\pfront} {p_{In}}
\newcommand{\tildepfront} {{\tilde p}_{In}}

\newcommand{\pback} {p_{Out}}
\newcommand{\Vpar}{V_\parallel}
\newcommand{\Vpara}{V_{\parallel,1}}
\newcommand{\Vparb}{V_{\parallel,2}}
\newcommand{\Vparc}{V_{\parallel,3}}

\newcommand{\eps}{\varepsilon}

\newcommand{\Jac}[1]{J_{#1}}

\global\def\Pt{-1}
\input{tikzdefines.tex}

\usepackage[nocompress]{cite}
\begin{document}

\title{Interactive Isogeometric Volume Visualization with Pixel-Accurate Geometry}

\teaser{

    \includegraphics[height=1.4in]{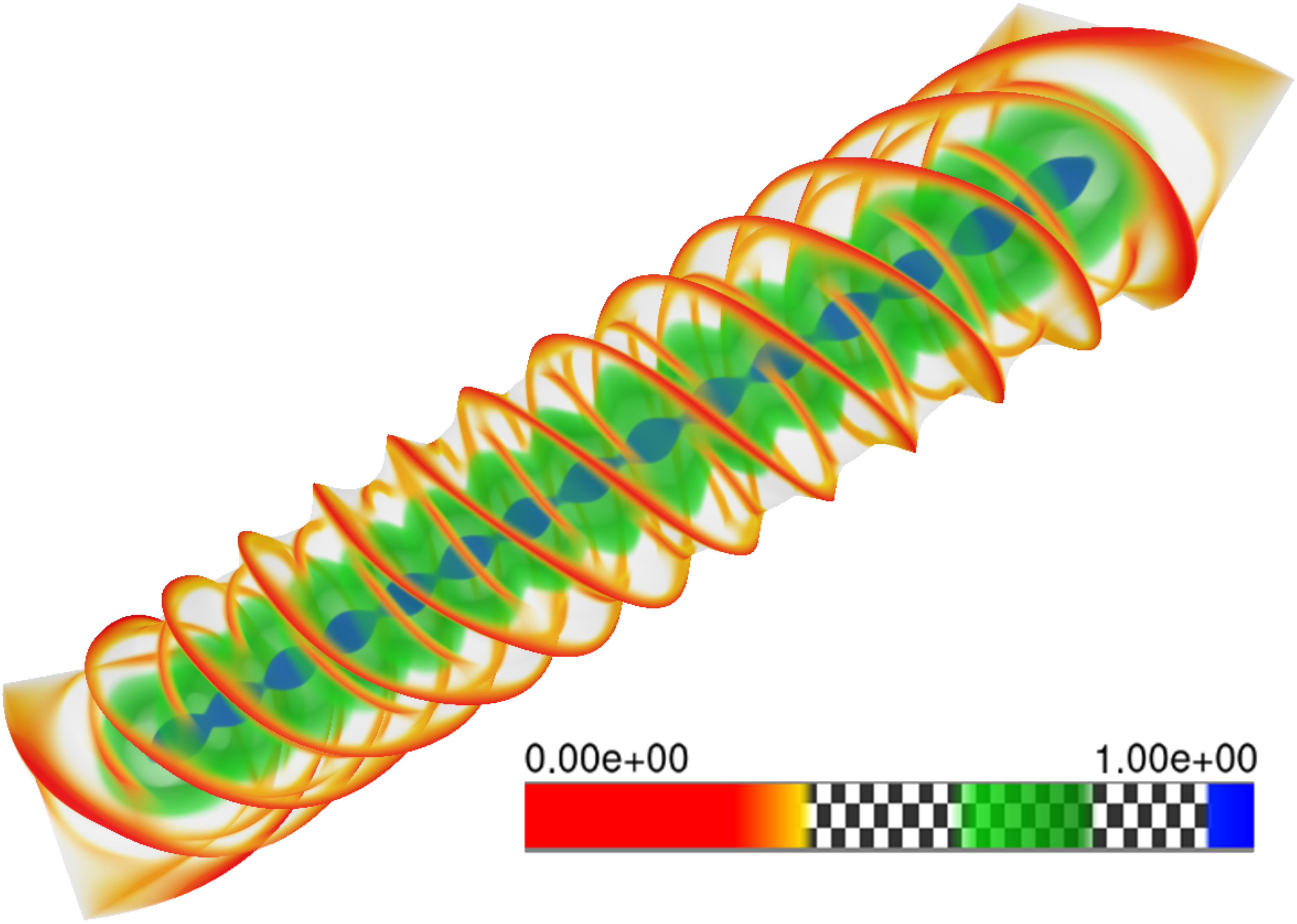}
    \includegraphics[height=1.3in]{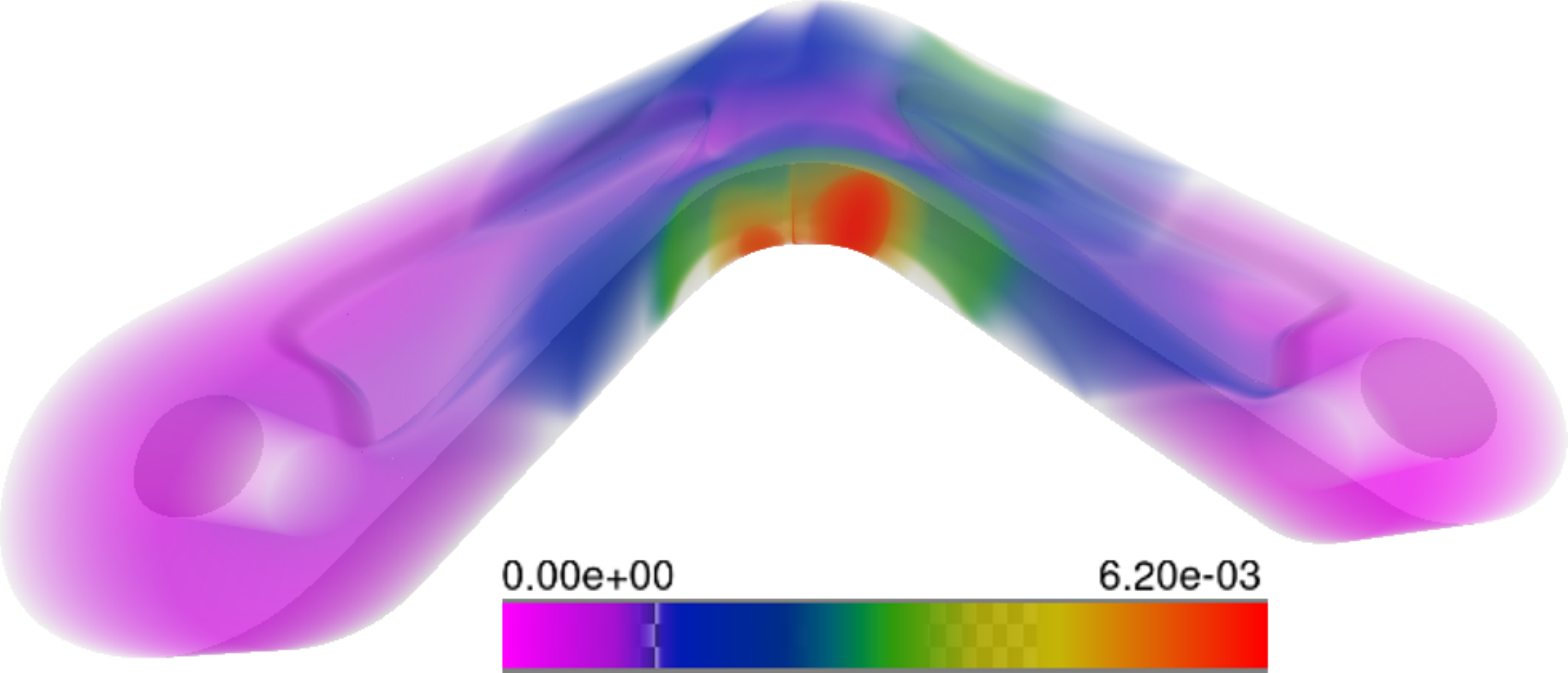}
    \includegraphics[height=1.3in]{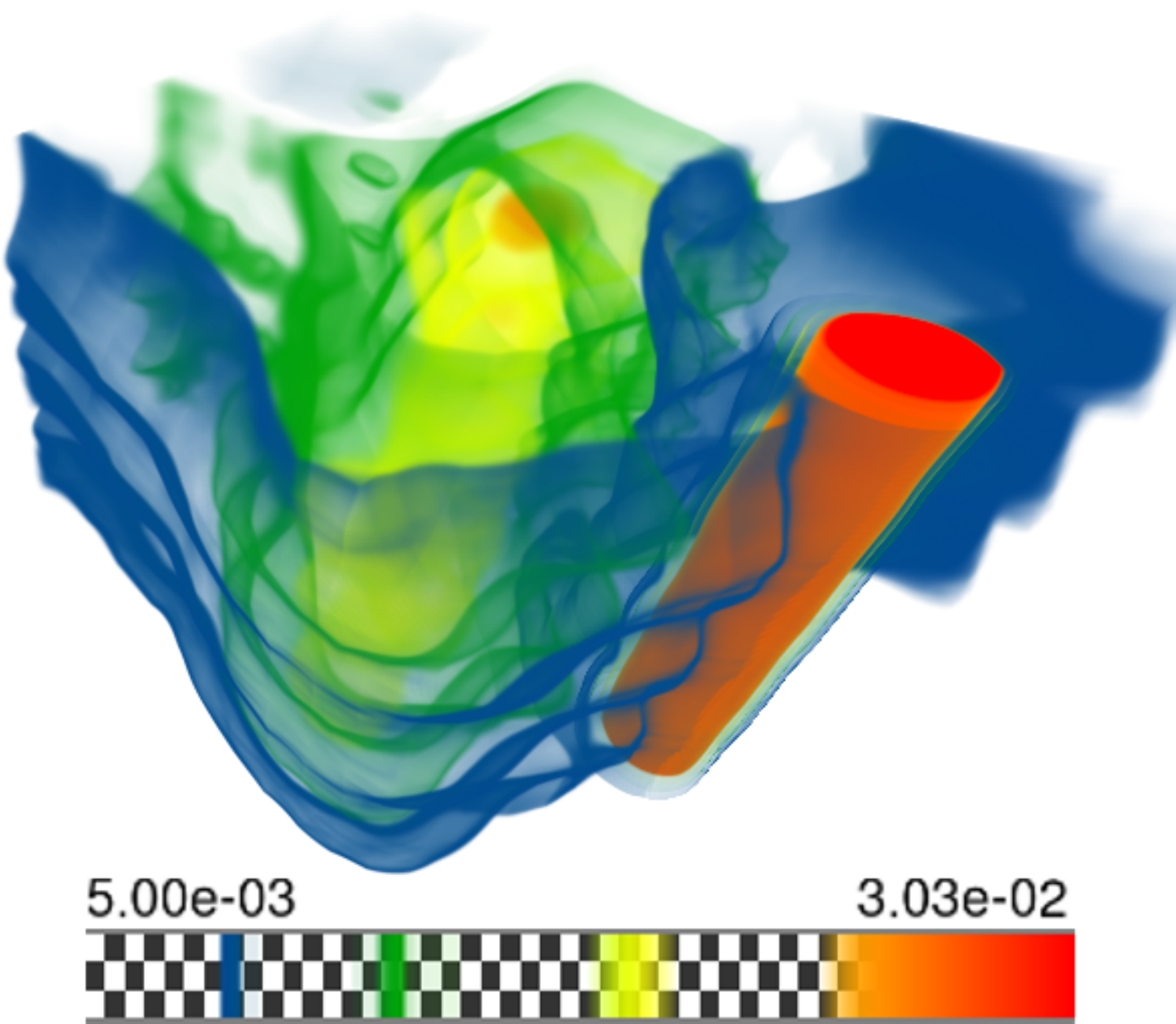}
    \captionof{figure}{Examples for isogeometric volume visualization in industry. Left to right: Twisted bar showing quality of parametrization; Industrial demonstrator model from TERRIFIC project showing von Mises stress of the bent model; Backstep flow from a computational fluid dynamics simulation showing turbulent viscosity.
        The bars show the colors that are assigned to the values of the scalar field, where a checkerboard pattern indicates transparent regions.
    }
    \protect \label{fig:teaser}
}

\author{
    Franz G. Fuchs,
    Jon M. Hjelmervik
    \IEEEcompsocitemizethanks{
    \IEEEcompsocthanksitem SINTEF ICT, Forskningsveien 1, N--0314 Oslo, Norway
    \protect\\
        E-mail: franzgeorgfuchs@gmail.com, jon.hjelmervik@sintef.no
    }
    \thanks{}
}

\IEEEoverridecommandlockouts
\IEEEpubid{
    \parbox{\textwidth}{\copyright 2015 IEEE, \url{http://dx.doi.org/10.1109/TVCG.2015.2430337}. Personal use of this material is permitted. Permission from IEEE must be obtained for all other users, including reprinting/ republishing this material for advertising or promotional purposes, creating new collective works for resale or redistribution to servers or lists, or reuse of any copyrighted components of this work in other works.}
}

%
\IEEEpubidadjcol

\IEEEcompsoctitleabstractindextext{%
    \begin{abstract}
    A recent development, called isogeometric analysis, provides a unified approach for design, analysis and optimization of functional products in industry.
    Traditional volume rendering methods for inspecting the results from the numerical simulations cannot be applied directly to isogeometric models.
    We present a novel approach for interactive visualization of isogeometric analysis results, ensuring correct, i.e., pixel-accurate geometry of the volume including its bounding surfaces.
    The entire OpenGL pipeline is used in a multi-stage algorithm leveraging techniques from surface rendering, order-independent transparency, as well as theory and numerical methods for ordinary differential equations.
    We showcase the efficiency of our approach on different models relevant to industry, ranging from quality inspection of the parametrization of the geometry, to stress analysis in linear elasticity, to visualization of computational fluid dynamics results.
    \end{abstract}

}

\maketitle

\IEEEdisplaynotcompsoctitleabstractindextext

\IEEEpeerreviewmaketitle

\section{Introduction}\label{sec:introduction}

Classic volume rendering is a method to display a two-dimensional projection of a three-dimensional scalar field that is discretely sampled on a Cartesian grid.
In order to achieve this, a model for radiative transfer is used to describe absorption and emission of light along view-rays.
This article extends classic volume rendering to isogeometric volumes,
where both geometry and scalar field are given in terms of splines (NURBS, B-splines, etc.).
We present a novel method for direct, interactive rendering of isogeometric models.
The efficiency and applicability of the proposed isogeometric volume rendering method is showcased in three different application areas relevant to industry, see \reffig{teaser}.

These types of models stem from isogeometric analysis (IGA), a recent development proposed by Hughes et al.~\cite{cottrell2009isogeometric} for the analysis of physical phenomena governed by partial differential equations.
IGA provides the integration of design and analysis by using a common representation for computer aided design (CAD) and finite element methods (FEM).
This eliminates the conversion step between CAD and FEM, which is estimated to take up to 80\% of the overall analysis time for complex designs \cite{cottrell2009isogeometric}.

The pipeline for design, analysis and optimization of functional products is depicted in \reffig{flow}.
Visualization is used in all the stages; for quality inspection of the geometry, for studying the results of the numerical analysis, and for marketing purposes.
It is therefore increasingly important to offer visualization techniques that are reliable, informative and visually pleasing.
A main advantage of IGA is that it enables the direct feedback from numerical analysis results to the CAD model.
However, for that process to work efficiently, it is essential to be able to interactively inspect the results from the numerical analysis stage.

The geometry of an isogeometric volume is given by a spline $\phi$, mapping each point of the parameter domain $P\subset\R^3$ to a point in the geometry domain $G\subset\R^3$, see \reffig{raycasting}.
In addition, a second spline $\rho$ 
is defined on the parameter domain $P$, describing a physical value such as density, displacement, temperature.
The spline $\rho$ can be scalar- or vector-valued and comes from a numerical simulation of a physical phenomenon.

In this paper we restrict our attention to B-splines. In three dimensions a B-spline of degree $p$ has the form
\begin{equation}\label{eq:Bspline}
    S(u,v,w) = \sum_{i=1}^l K_i^p(u) \sum_{j=1}^m L_j^p(v)\sum_{k=1}^n M_k^p(w) C_{i,j,k},
\end{equation}
where $C_{i,j,k}\in\R^d$ are the control points defined over a set of non-decreasing knot vectors $U=\{u_1,...,u_{l+p+1}\}$, $V=\{v_1,...,v_{m+p+1}\}$ and $W=\{w_1,...,w_{n+p+1}\}$.
With $K_i^p,L_i^p,M_i^p$ we denote the recursively defined $i$-th B-spline of degree $p$ in the corresponding direction.
For a scalar spline $\rho(u,v,w)$ the control points $C_{i,j,k}$ are scalar valued, and for the spline $\phi(u,v,w)$ describing the geometry the control points have values in $\R^3$.

\begin{figure}
\begin{tikzpicture}[scale=0.7]
        \path (0.5,8.7) node(b) {\includegraphics[width=0.30\linewidth]{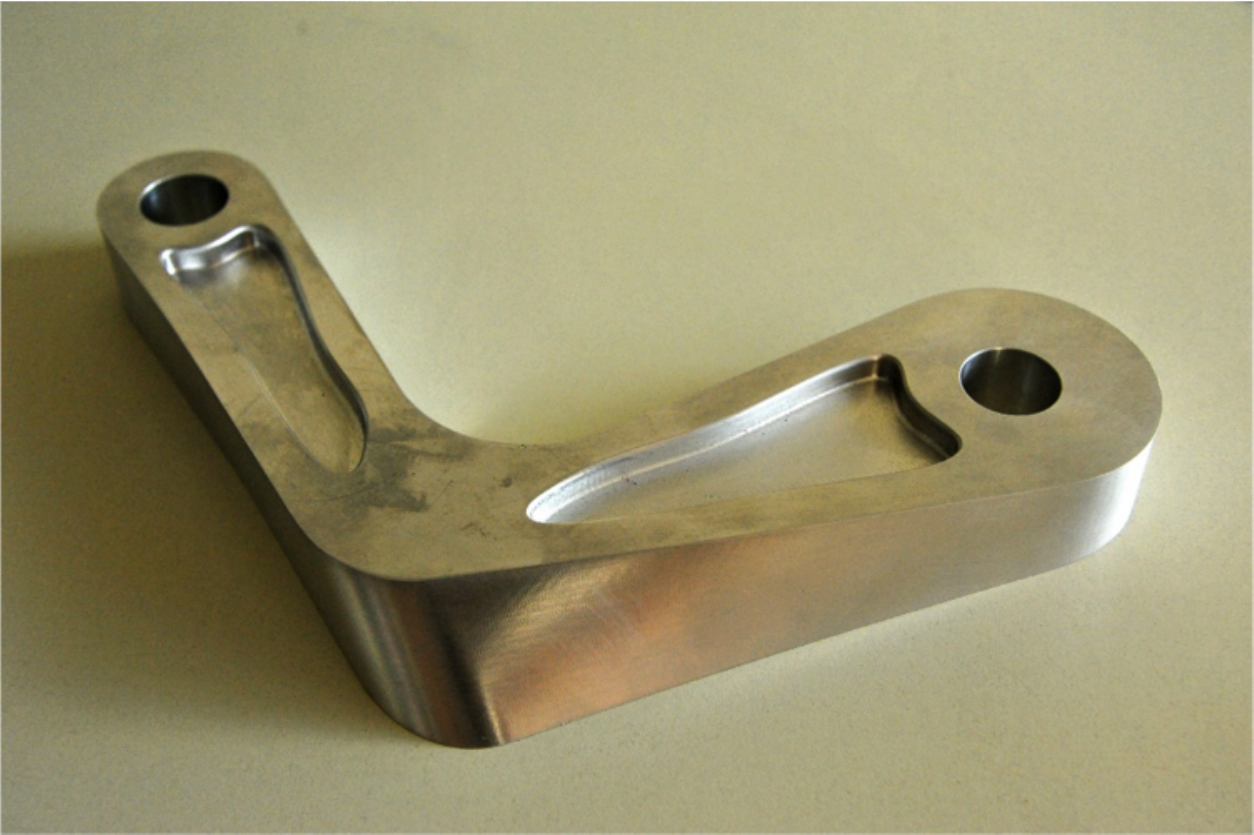}};
        \path (5,7.5) node(c) {\includegraphics[width=0.30\linewidth]{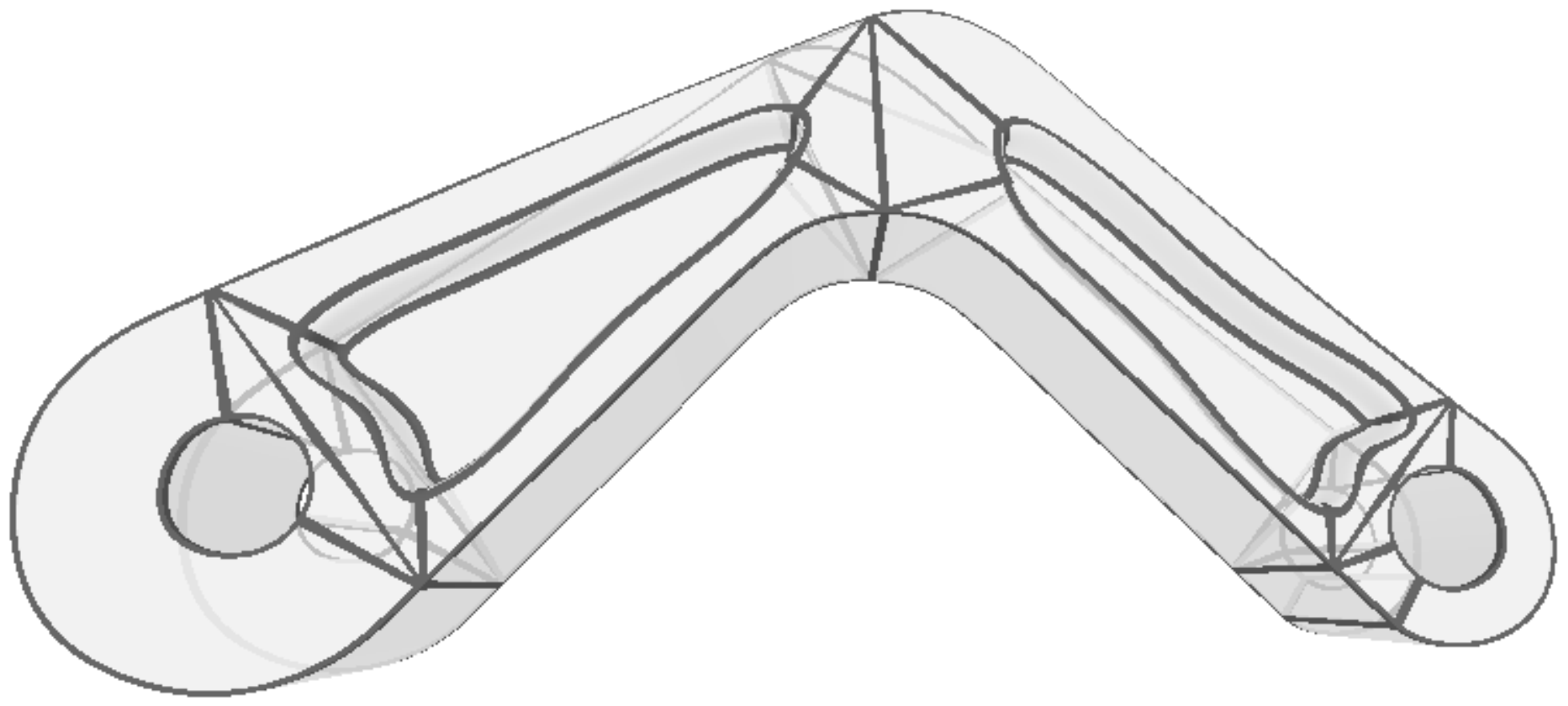}};
        \path (8.5,3) node(d) {\parbox[c][4\baselineskip][c]{2.6cm}{
        $F(x_1,..,x_n,$\\
        $\phantom{F(}\frac{\partial u}{\partial x_1},..,\frac{\partial u}{\partial x_n},$\\
        $\phantom{F(}\frac{\partial^2 u}{\partial x_1 \partial x_1},..) = 0$}};
        \path (1.5,3) node(e) {\includegraphics[width=0.50\linewidth]{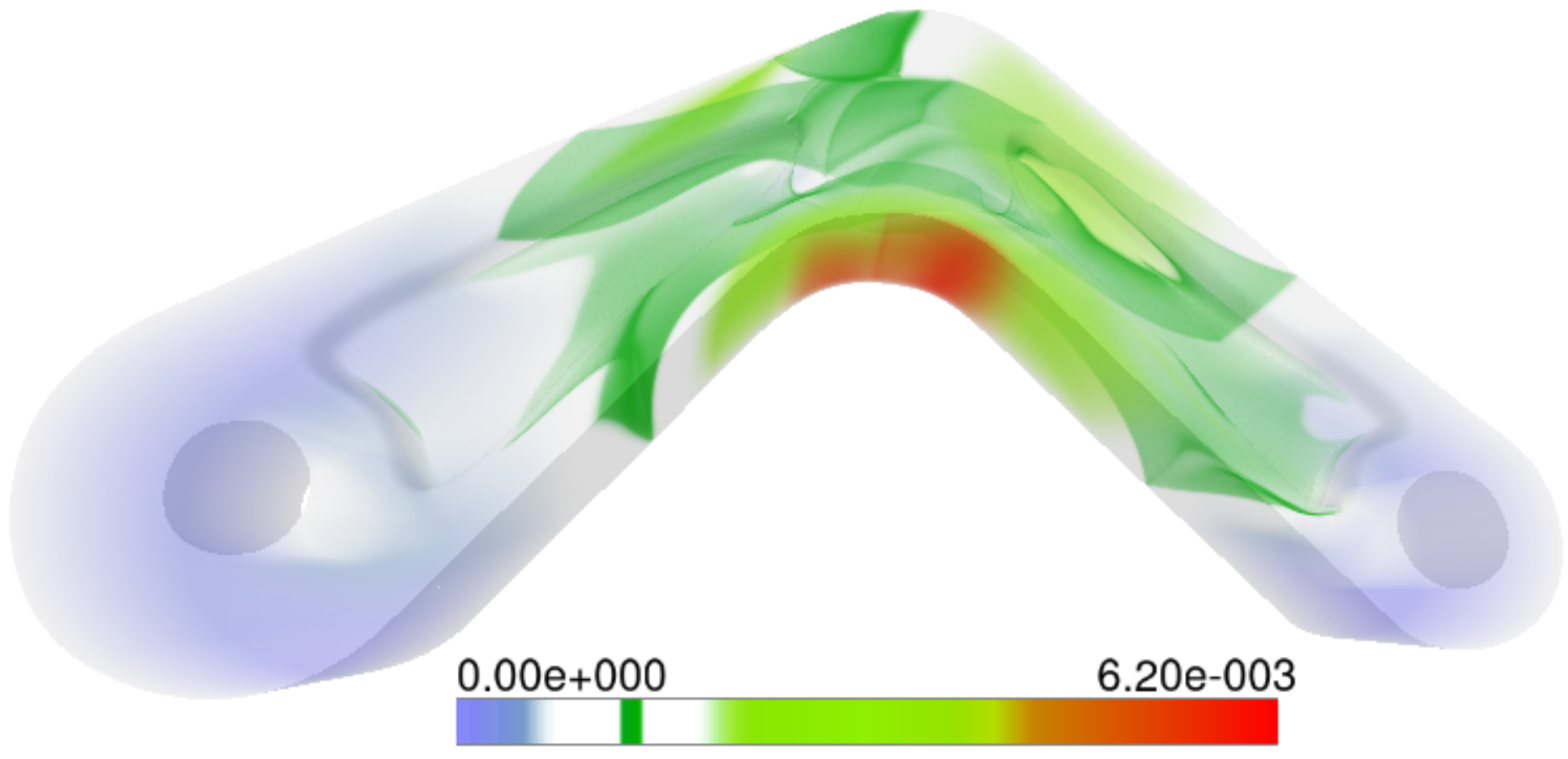}};
        \draw[very thick, bend left=45, <->] (b) to node [midway,above, rotate=-35]{model}(c) ++ (b) to node [midway,below, rotate=-35]{produce}(c);
        \draw[very thick, bend left=45, ->] (c) to node [midway,above, rotate=-60]{simulate} (d);
        \draw[very thick, bend left=45, ->] (d) to node [midway, below]{visualize/analyze} (e);
        \draw[very thick, bend left=40, ->] (e) to node [midway, above, rotate=55]{update/adjust} (c);
        \node (same) at (5,5) {same representation};
        \draw[very thick, -triangle 45] (same) -- (c);
        \draw[very thick, -triangle 45] (same) -- (d);
        \draw[very thick, shorten >=-20pt, -triangle 45] (e) -- (same);
\end{tikzpicture}    
    \caption{Information flow of an isogeometric object.}
    \protect \label{fig:flow}
\end{figure}

\section{Related Work}\label{sec:relatedwork}
Scientific volume visualization techniques convey information about a scalar field defined on a given geometry.
The techniques can be divided into the following approaches: simply rendering the bounding surfaces of the object; iso-surface extraction; and volume rendering.
The main challenges for achieving interactive volume visualization in the setting of isogeometric volumes are that an explicit expression of the inverse function of the geometry is not available in general, and that sampling is computationally expensive due to the need for spline evaluation, i.e., piecewise polynomial functions.

The first approach for volume visualization is to render the scalar field on the outer surfaces of the isogeometric volume only.
Although no information in the interior can be retrieved, this is a popular method due to its low computational effort.
Methods based on ray-casting parametric polynomial surfaces are a well-studied but challenging problem, see e.g. Kajiya~\cite{Kajiya:1982:RTP}.
B-spline surfaces can be rendered as piecewise algebraic surfaces, see for instance Loop and Blinn \cite{Loop:2006:RGR}.
But finding corresponding scalar field values becomes difficult, because algebraic surfaces are not parametrized.
An alternative to ray-casting surfaces is rasterization, in particular with the recent introduction of the tessellation shader stage in graphics processing units (GPUs).
A GPU-based two-pass algorithm for pixel-accurate rendering of B-spline surfaces was presented by Yeo et al. \cite{Yeo:2012}:
The first pass determines a sufficient tessellation level for each patch; during the second pass the surface is actually tessellated.
An alternative method which is also pixel-accurate, is presented by Hjelmervik~\cite{hjelmervik2012direct}, where bounds on the second order derivatives decide sufficient tessellation levels, without querying neighboring patches, allowing a single-pass algorithm.

\begin{figure}
\centering
    \begin{tikzpicture}[
	scale=0.72,
	grid/.style={very thin,gray},
	cube/.style={thick,line join=round},
	cube hidden/.style={thick,dashed,line join=round}]
    \global\def\translatePAR{5.2}
    \global\def\translateRHO{0.5}
	\drawRayCastingISOsetup
	\drawRayCastingISOscreen
    \node (pixel) at ({\translateGEO+\Axhalf-.1},{\Ayhalf},{\Azhalf-.2}) {};
    \node (pixeldescription) at ({\translateGEO+\Axhalf+2.5},{\Ayhalf},{\Azhalf-.2}) {pixel};
    \draw[very thick,-triangle 45] (pixeldescription) to node {} (pixel);
	\drawRayCastingISOrayingeoSurface
	\drawScalarField
	\drawRayCastingISOrayinparamVolume
	\drawRayCastingISOrayinparamSurface
	\drawRayCastingISObase
    \annotateInOut
	\node (S) at (\translateRHO+\translatePAR+3,2.5) {$\mathbb{R}$};
    \draw[thick,-triangle 45] (P) to node[above] {$\phi$} (G);
	\draw[thick,-triangle 45] (P) to node[above] {$\rho$} (S);
    \end{tikzpicture}
    \caption{Isogeometric volume rendering: $\phi$ describes the geometry, $\rho$ defines a scalar field, both defined on $P$.}
    \protect \label{fig:raycasting}
\end{figure}
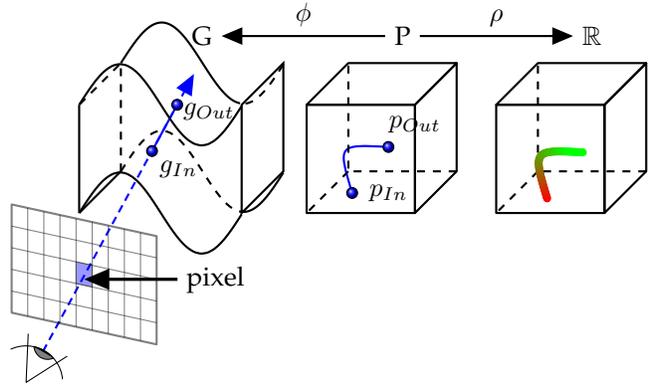

The second approach is to display iso-surfaces, i.e., surfaces where the scalar field has a particular value.
If the scalar field is sampled discretely over a regular grid, the marching cubes algorithm (see Lorensen and Klein~\cite{Lorensen:1987}) provides an efficient implementation; see e.g., Dyken et al.~\cite{Dyken:2008} for an implementation on the GPU.
However, for isogeometric models, where both the geometry and the scalar field are given by spline functions, the marching cubes algorithm cannot be applied directly.
Martin and Cohen~\cite{Martin_Cohen:2001:REVAUTS} provide an algorithm for iso-surface extraction in the setting of isogeometric volumes.
The method consists of iteratively dividing the model into a set of \Bezier volumes.
When those volumes are sufficiently simple, the iso-surfaces are given as the root of a function and the Newton-Raphson method is used to approximate the surfaces.
This framework has been realized as a CPU-based parallel implementation (see Martin et al.~\cite{Martin:2012}). The reported timings between 1 and 7~frames per second (FPS) on a cluster, make the approach unsuitable for interactive visualization purposes.

Recently, Schollmeyer and Fr\"{o}hlich\cite{schollmeyer2014direct} presented a GPU-based multi-pass visualization technique for direct iso-surface ray casting of NURBS-based isogeometric volumes.
The first pass generates a list of ray intervals which potentially contain intersections with the faces of each \Bezier cell.
After applying culling and depth-sorting, this list is used to generate ray-surface intersections in the second pass.
The ray-surface intersections are given by the roots of a system of nonlinear equations in the third pass.
An elaborate root isolating method is applied to find all ray-surface intersections.
A GPU-based implementation shows interactive volume visualization results of their method.

Another approach is to model the scalar field as a participating medium, where a modifiable transfer function specifies how field values are mapped to emitted color and transparency.
If the field consists of discrete samples over a regular grid, an abundance of results is available, see e.g. Levoy~\cite{Levoy:1988:DOS} for an early example or Engel at al.~\cite{engel2006real} for an overview.
In order to make use of existing standard volume rendering methods, one possibility is to precompute a "voxelized" version of the isogeometric model:
Taking the geometry into account, one can store the values of the scalar field in a texture (the algorithm proposed in this article can readily be used for that).
However, there are several draw-backs of this approach:
Firstly, as illustrated in \reffig{outerSurfaces}, it is difficult to represent the outer surfaces of the isogeometric volume using a voxel grid.
Isogeometric objects typically have a smooth outer surface, representing a sharp transition between where a scalar field is defined and the outside.
Standard volume rendering will therefore typically lead to spurious block-like structures.
Secondly, volume rendering of a voxelized model will (tri-) linearly interpolated the sample points.
A linear interpolation $p_1(x)$ of a function $f(x)$ between two sample points $a,b$ has the following (optimal) error bound
$|f(x)-p_1(x)| \leq \frac{(b-a)^2}{8} \underset{x\in[a,b]}{\max} \|f''(x)\|$, see e.g., \cite{hjelmervik2012direct}.
A similar bound holds in 2 and 3 dimensions.
In IGA the scalar field is given by $f(x,y,z)=\rho(\phi^{-1}(x,y,z))$, where $\rho$, and $\phi$ are spline functions (see \reffig{raycasting}).
This means that a point (such as the red point in \reffig{outerSurfaces}~(d)) that is (tri-) linearly interpolated between sample points can have an arbitrarily large approximation error, depending on the second order derivatives of $f(x,y,z)$.
We refer also to \reffig{deltaEscaled} showing how this affects render quality.
As a consequence, a voxelized version needs a potentially very high number of voxels in order to represent the scalar field accurately.
The high demand on GPU memory decreases the efficiency of standard volume rendering algorithms, see \refsec{applications}.
Of course, more sophisticated methods with adaptively chosen sample points could be used to represent the scalar field.
However, since a main reason for the introduction of IGA was to enable geometrically exact representation of geometry, this advantage should not be lost in the visualization stage, see \reffig{flow}.
Exact geometry is desired by designers and analyzers alike.

An algorithm for direct volume rendering for freeform volumes was presented by Chang et al.~\cite{Chang:1995:DRO}.
The method consists of subdividing B-spline volumes into a set of \Bezier volumes until the geometry of the volumes is monotone.
Those volumes are then depth-sorted, and scan-converted, allowing direct blending.
The algorithm reaches approximately 5~FPS on a mini-supercomputer with 8192 processors.
Martin and Cohen~\cite{Martin_Cohen:2001:REVAUTS} outline an algorithm based on finding the roots of a function with the Newton-Raphson method.
However, to the best of our knowledge, there is no implementation of this algorithm for isogeometric models.

\begin{figure}
    \centering
    \subfigure[Proposed Method]{\includegraphics[width=0.15\textwidth]{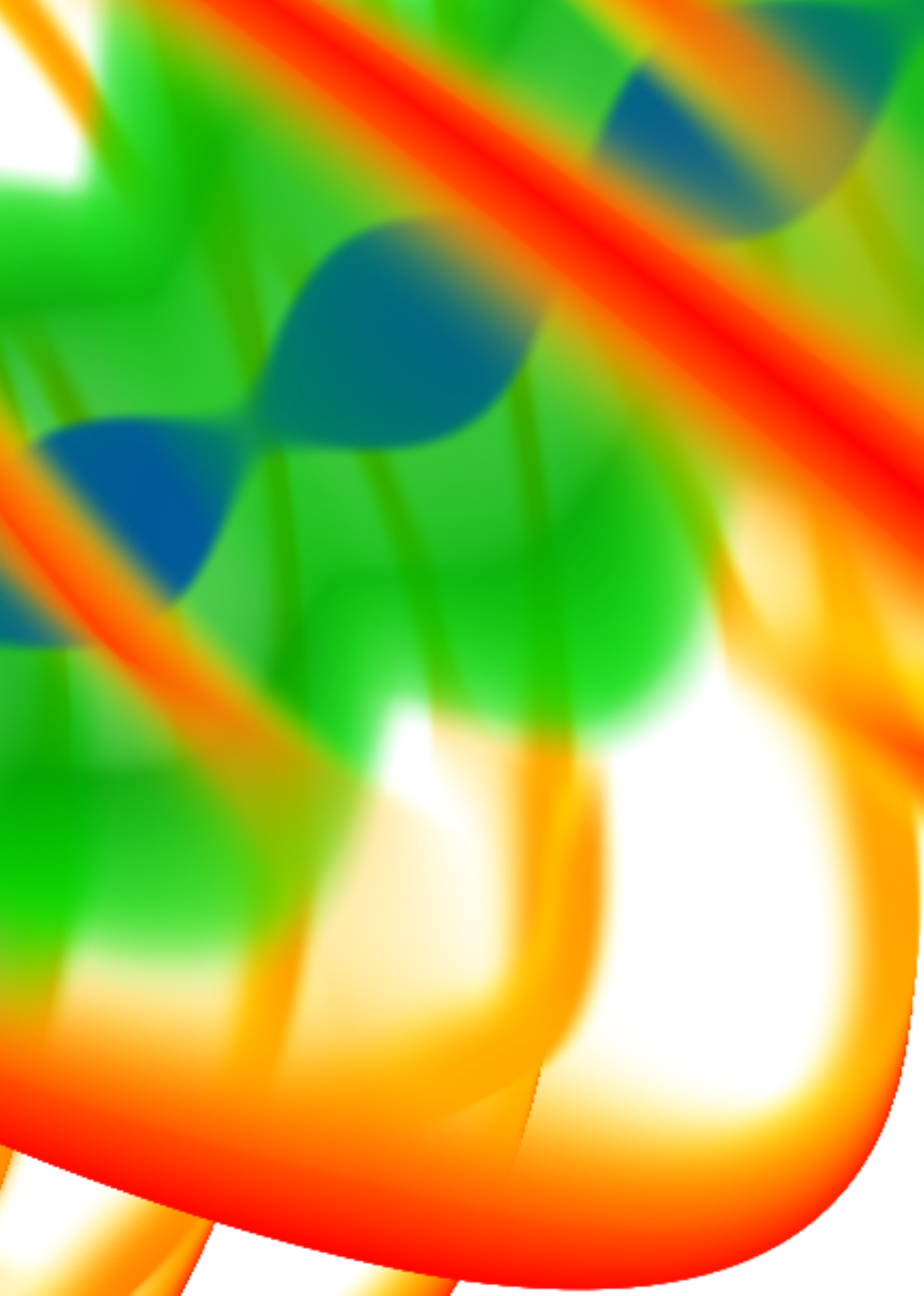}}
\hspace{1cm}
    \subfigure[Voxelized Method]{\includegraphics[width=0.15\textwidth]{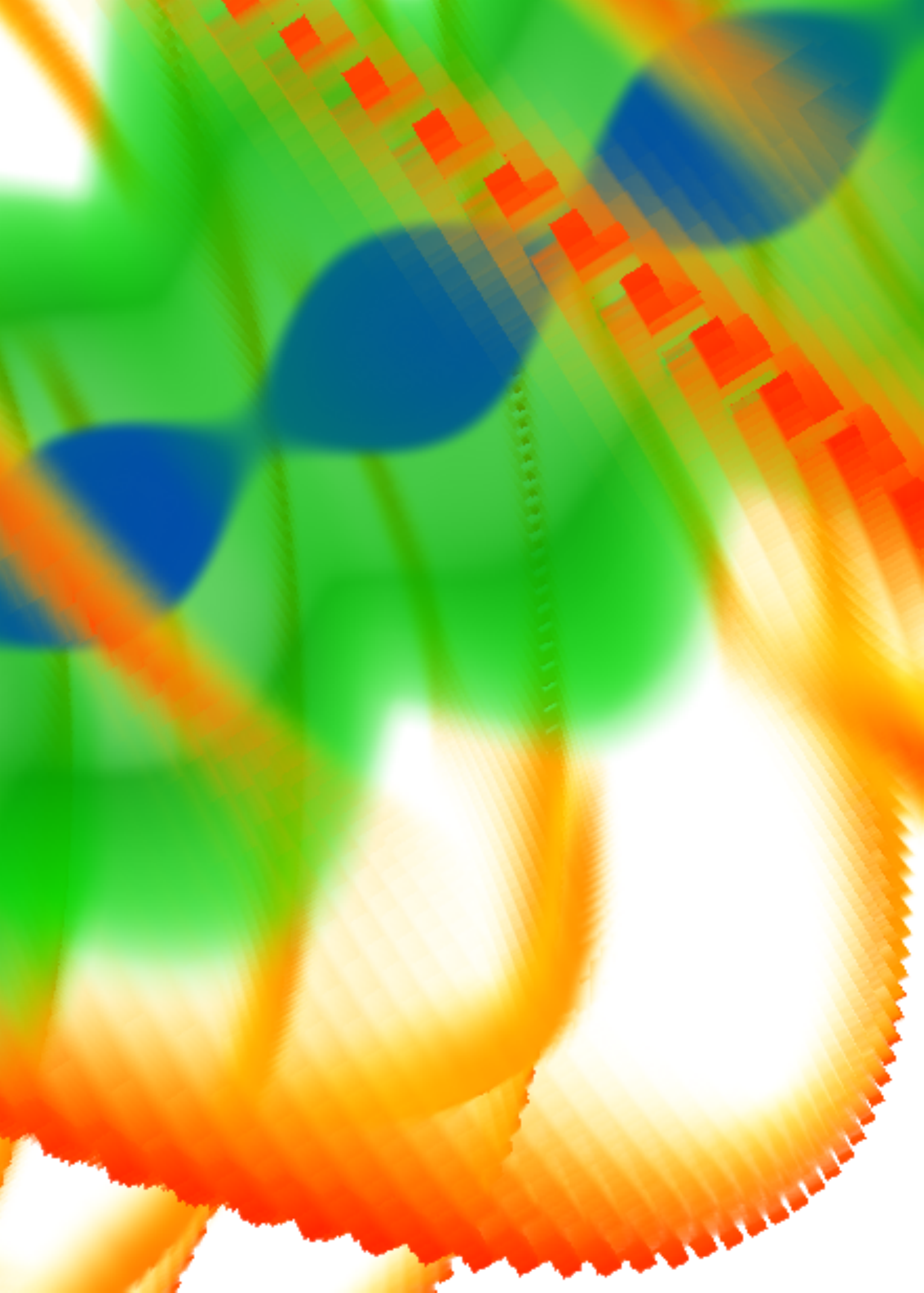}}
    \subfigure[
        Schematic sketch of boundary treatment for the proposed volume rendering method.
        Dashed blue line is pixel-accurate approximation.
    ]{
        \begin{tikzpicture}[
            scale=0.7,
            grid/.style={very thin,gray},
            cube/.style={thick,line join=round},
            cube hidden/.style={thick,dashed,line join=round}]
            \drawProposedSurf
        \end{tikzpicture}
    }
    \subfigure[
        It is difficult to exactly represent the surface of the green object with the grid values shown in blue (schematic representation).
    ]{
        \begin{tikzpicture}[
            scale=0.7,
            grid/.style={very thin,gray},
            cube/.style={thick,line join=round},
            cube hidden/.style={thick,dashed,line join=round}]
            \drawVoxSurf
        \end{tikzpicture}
    }
    \caption{
        The proposed method leads to pixel-accurate outer surfaces, while standard volume rendering of a voxelized version of the model is prone to show block-like structures at the outer surfaces.
        The red sample point in (d) will be (tri-) linearly interpolated.
    }
    \protect \label{fig:outerSurfaces}
\end{figure}

In \cite{Kurzion95spacedeformation} Kurzion and Yagel present a method for visualizing deformed two- and three-dimensional models.
Instead of deforming the objects themselves, the geometry is given by so-called deflectors that bend the rays used to render the scene.
However, in contrast to isogeometric volumes, the geometry is given explicitly. 

In this paper we present a flexible framework for volumetric visualization based on volume rendering, allowing visualization of scalar fields as well as derived properties such as parametrization quality or mechanical stress.
Our approach consists of several stages, leveraging the strengths of existing algorithms where possible.
Several features of our approach are novel:
\begin{enumerate}
    \item We provide an extremely robust and efficient algorithm for determining view-ray intersections with the surfaces of the volume, see \refsec{surfacerendering}.
          This is achieved by reformulating the original problem of finding zeros of a function, to be a problem related to approximation of surfaces.
      \item We devise novel approaches for pixel-accurate approximation of the preimage (i.e., the inverse) of the view-ray in the parameter space, suitable for efficient implementation on modern GPUs, see \refsec{geometry}.
          In addition to an algorithm based on approximating zeros of a function, we provide an alternative based on ordinary differential equations (ODEs).
    \item Degenerate cases of the parametrization of the geometry are treated in a suitable way, further increasing the robustness of our approach, see \refsec{cornercases}.
\end{enumerate}

We would like to point the reader to \cite{deuflhard2002scientific} for a good introduction to numerical methods for ODEs as well as \cite{chorin1990mathematical} for a mathematical introduction to fluid mechanics.

The rest of the paper is organized as follows: \refsec{volumerendering} provides the necessary background for volume rendering of isogeometric models, followed by a description of our approach in \refsec{approach}.
Then, novel algorithms enabling geometrically pixel-accurate sampling of the volume render integral~\eqref{eq:volRenderIntegral} are described in \refsec{geometry}.
Finally, we present applications and provide details of the performance of the implementation of the overall algorithm in \refsec{applications}, and a conclusion in \refsec{conclusion}.

\section{Volume Rendering for Isogeometric Models} \label{sec:volumerendering}

In this article we present an algorithm for volume rendering based on tracing view-rays through the volume from an imaginary observer.
If such a view-ray intersects the object one obtains the color for the pixel of the screen by evaluating an integral describing the accumulated radiance along the ray.
For a more detailed description of well established techniques for volume rendering see Engel et al.~\cite{engel2006real}, Jensen~\cite{Jensen:2001:RISUPM} and references therein.

\subsection{Continuous Model}
When taking both emission and absorption into account, the accumulated radiance $I_\lambda$ for wave length $\lambda$ along a view-ray $\gamma:\R\rightarrow\R^3$ is given by the so-called volume render integral
\begin{equation}
    \label{eq:volRenderIntegral}
    I_{\lambda}(t) = I_{\lambda}(0) T_{\lambda}(0,t) + \int_{ \left.\gamma\right|_{[0,t]} } \!\!\! \!\!\! \glow_{\lambda}(s) T_{\lambda}(s,t) ds,
\end{equation}
where $\int_\gamma$ denotes the line integral.
The function $\glow_\lambda(s)$ specifies emission, and $T_\lambda(s,t)$ specifies absorption (from $s$ to $t$) of light with the wave length $\lambda$.
In applications, one typically uses three groups of wave lengths representing red, green, and blue.
The emission and absorption, defined by a so-called transfer function, depend on the value of the scalar field $\rho$.

A major difference to classic volume rendering is that the geometry is no longer trivial.
For isogeometric models, both the spline $\phi(u,v,w)$ describing the geometry as well as the scalar field $\rho(u,v,w)$, are defined on the same parameter domain $P$, see \reffig{raycasting}.
This means that the value of the scalar field along the view-ray $\gamma(s)$ is given by
\begin{equation}
    \label{eq:rhoPhiInv}
    \rho_\gamma(s) \coloneqq \rho(\phi^{-1}(\gamma(s))).
\end{equation}
If not otherwise stated, $\phi$ is assumed to be bijective.
As mentioned before, both emission and absorption in the volume render integral~\eqref{eq:volRenderIntegral} are functions of the values of the scalar field, i.e.,
\begin{equation}
    \label{eq:emissionabsorption}
    \glow(s) = \glow(\rho_\gamma(s)), \quad T(s,t) = T(\rho_\gamma(s), \rho_\gamma(t)).
\end{equation}

As a consequence, if $\phi$ is not linear, the value of the scalar field along the straight view-ray in the geometry domain $G$, is obtained along a (not straight) curve in the parameter domain $P$.

\subsection{Numerical Approximation}\label{sec:numapprox}
In general, neither the inverse of the spline $\phi$ (see \refequ{rhoPhiInv}) nor the volume render integral~\eqref{eq:volRenderIntegral} itself have closed-form solutions.
Therefore, the solutions have to be approximated and the overall error will consist of different sources due to
\begin{itemize}
    \item numerical quadrature of integral~\eqref{eq:volRenderIntegral} (depending on number of sample points and their location), and
    \item numerical approximation of the preimage of these sample points along the view-ray (method depended).
\end{itemize}
For interactive applications, a compromise between performance and accuracy must be found.
The numerical quadrature of the volume render integral is described in \refsec{quadrature}.
The approximation of the inverse is specific to isogeometric models, and the main contribution of this paper is dedicated to it (sections~\ref{sec:approach} and \ref{sec:geometry}).
A definition of the requirement for the accuracy of the approximation of the inverse is given in \refsec{pixelaccuracy}.

\begin{figure}[t]
    \centering
    \begin{tikzpicture}[scale=0.45,line width=0.3mm]
        \drawTrans
        \draw[blue,domain=-1.5:1] plot[smooth] ({\datat(\x,\snull+\dist)},{\x+1.5});
        \drawSubsample
        \drawTransLR
        \drawNormal
        \draw[blue,domain=-1.5:1] plot[smooth] ({\datat(\x,\snull+\distNormal)},{\x+1.5});
    \end{tikzpicture}
    \caption{
        Supersampling: Dense sampling of high frequencies of the transfer function $c$ can improve the approximation (shown in red) of the volume render integral~\eqref{eq:volRenderIntegral} without further evaluations of $\rho_\gamma$.
        }
    \protect \label{fig:subsample}
\end{figure}

\subsubsection{Quadrature of volume render integral}\label{sec:quadrature}
Discretizations of \refequ{volRenderIntegral} are based on splitting the integral into intervals.
Efficient implementations on GPUs are so-called compositing schemes where color and opacity is accumulated iteratively.
Front-to-back compositing 
for the accumulated radiance $C_{dst} = (I_r,I_g,I_b)^T$, and the accumulated opacity $\alpha_{dst} = (1-T_{dst})$
is given by \refalg{compositing}.

\begin{algorithm}
    \caption{Front-to-back compositing}
    \label{alg:compositing}
\begin{algorithmic}
    \State $T \gets (1-\alpha_{src})^{\Delta s_i/\xi}$
    \State $C_{dst} \gets C_{dst} + (1-T) (1-\alpha_{dst}) C_{src}$
    \State $\alpha_{dst} \gets\, \alpha_{dst} + (1-T) (1-\alpha_{dst})$
\end{algorithmic}
\end{algorithm}

\noindent
Here, $\Delta s_i$ is the varying ray segment length and $\xi$ is a standard length.
Furthermore, $C_{src}$ and $\alpha_{src}$ are given by the transfer function through \refequ{emissionabsorption}.

In many application areas transfer functions contain high frequency components, dictating a high sampling rate (Nyquist rate).
One method is \textit{oversampling}, i.e., introducing additional sampling points, although the underlying scalar field is approximately linear.
Evaluating the scalar function in the setting of isogeometric volume rendering is a time-intensive operation as it means computing \eqref{eq:Bspline}.

To avoid oversampling a common technique is \textit{pre-integration} (see e.g., \cite{roettger2003smart,Engel:2001}), which is based on calculating the volume render integral for pairs of sample values in advance.
Although this approach can successfully be applied in many cases, it has the flaw that it only works for equidistant sample points, because the volume render integral ~\eqref{eq:volRenderIntegral} is nonlinear.

To avoid oversampling, but account for the non-linearity of $I$, we use the following technique called \textit{supersampling}.
In contrast to pre-integration, supersampling uses a dense quadrature of the volume render integral assuming linearity of the scalar field between two sample points $\rho_\gamma(s_j),\rho_\gamma(s_{j-1})$ (not assuming linearity of the volume render integral between sample points), see also \reffig{subsample}.

Although this approach also has its weaknesses, it is easy to implement, and it successfully captures high frequencies of the transfer function.
This increases the image quality by reducing wood-grain artifacts, while avoiding computationally expensive evaluations of the spline function $\rho$ and $\phi$.

\subsubsection{Pixel-Accurate Rendering of Geometry}\label{sec:pixelaccuracy}

\begin{figure}[t]
    \centering
        \begin{tikzpicture}[
        scale=0.8,
        grid/.style={very thin,gray},
        cube/.style={thick,line join=round},
        cube hidden/.style={thick,dashed,line join=round}]
        \drawRayCastingISOsetup
        \def\translatePAR{6.5}
        \drawRayCastingISOscreen
        \drawFrustum
        \drawRayCastingISOrayinparamVolume
        \drawRayCastingISObase
        \drawBadApprox
        \node (frustum) at (0.2,-0.5) {};
        \node (frustumdescription) at (2.2,-2.5) {pixel frustum};
        \draw[thick,-triangle 45] (frustumdescription) to node {} (frustum);
        \end{tikzpicture}
    \caption{Pixel-accuracy: The preimage of the view-ray (blue curve in P) is approximated by the red points $\widetilde p_i$ in P. The image $\widetilde g_i = \phi(\widetilde p_i)$ should lie inside the pixel frustum, shown as a blue box in G, and have the correct depth order.}
    \protect \label{fig:pixelaccurate}
\end{figure}

In order to display the correct geometry of the isogeometric object, the approximation of the inverse of the geometry spline $\phi$ has to meet the following requirement, see \reffig{pixelaccurate} for an illustration.
Following \cite{Yeo:2012,hjelmervik2012direct}, we define the following.
\begin{definition}[Pixel-Accurate Approximation of Geometry]\label{def:pixelaccuracy}
    Let $\gamma(s)$ be the view-ray for a given pixel on the screen.
    The points $\tilde p_i, 1\leq i\leq n$ in the parameter domain P are called {\normalfont pixel accurate} sample points, if the following two requirements are fulfilled.
    \begin {itemize}
        \item All points $\tilde p_i$ must project into the pixel of the view-ray ({\normalfont "parametric accuracy"}):
            \begin{equation}
                \label{eq:deltaP}
                \Delta P \coloneqq 2 \Big\| \pi_s(\phi(\tilde p_i)) - {\scriptsize \begin{bmatrix} x \\ y \end{bmatrix}}\Big \|_\infty \leq 1,
            \end{equation}
            where ${\scriptsize \begin{bmatrix} x \\ y \end{bmatrix}}$ is the pixel's center and $\pi_s:\R^3\rightarrow\R^2$ is the projection to the screen.
        \item All points $\tilde p_i$ must have correct depth ordering along the view-ray ({\normalfont "covering accuracy"}).
For $2\leq i\leq n$:
            \begin{equation}
                \| \pi_\gamma(\phi(\tilde p_i)) - g_\text{eye} \|_2 \geq \| \pi_\gamma(\phi(\tilde p_{i-1})) - g_\text{eye} \|_2,
            \end{equation}
            where $g_\text{eye}$ is the position of the observer, and $\pi_\gamma:\R^3\rightarrow\R^3$ is the orthogonal projection onto the view-ray.
    \end{itemize}
\end{definition}
\noindent
This definition allows for a variable sample distance.

\section{Approach and Implementation}\label{sec:approach}

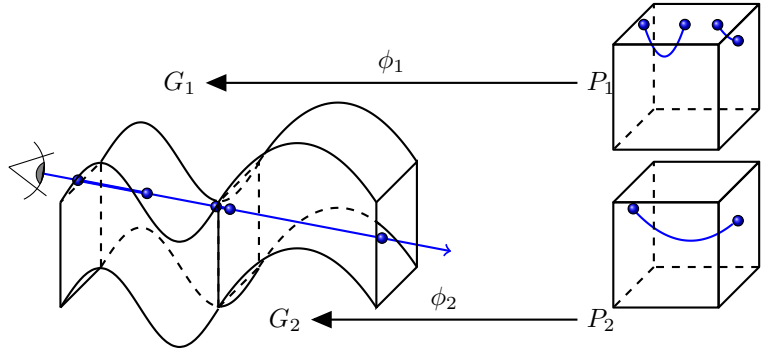
\begin{figure*}
    \centering
        \begin{tikzpicture}[
        scale=0.7,
        grid/.style={very thin,gray},
        cube/.style={thick,line join=round},
        cube hidden/.style={thick,dashed,line join=round}]
        \def\translatePAR{6.5+3}
        \drawRayCastingISOsetup
        \drawRayCastingISOrayingeoVolumeDepth
        \def\ma{3}
        \def\mb{0}
        \drawRayCastingISOrayinparamVolumeDepthOne{\ma}
        \drawRayCastingISOrayinparamVolumeDepthTwo{\mb}
        \drawRayCastingMultiblock{\ma}{\mb}
        \end{tikzpicture}
    \caption{
        This isogeometric model consists of two volume blocks $G_1$ and $G_2$.
        Both non-convexity and the multi-block structure increase the number of possible intersections per view-ray.
        }
    \protect \label{fig:depthsorting}
\end{figure*}

We present a novel approach enabling interactive volume visualization of isogeometric models with pixel-accurate geometry.
It is often necessary to partition the model in order to be able to model real-life features such as holes or to improve the quality of the parametrization.
Therefore, isogeometric models often consist of a number of so-called \textit{volume blocks} $G_\alpha$, $1\leq\alpha\leq a$, see \reffig{depthsorting}.
Each block has an corresponding spline $\phi_\alpha: P_\alpha \rightarrow G_\alpha$ defining the geometry, and scalar field $\rho_\alpha: P_\alpha \rightarrow \R$.

As described in \refsec{volumerendering}, an important step for approximating the volume render integral~\eqref{eq:volRenderIntegral} is to determine the intersections of the view-ray with the surfaces of the object.
When designers create isogeometric objects they often collapse edges or align two faces, which leads to singularities of the Jacobian of $\phi_\alpha$ on the boundary.
Existing approaches often struggle with finding intersections, leading to computationally expensive algorithms.
In order to devise an efficient and stable algorithm, our approach for volume-rendering consists of the following stages that are executed for every frame as part of the render pipeline.

\begin{enumerate}
    \item \textit{View-Ray intersections with surfaces:}
          We determine the intersections of all view-rays with the surfaces of the object by reinterpreting the problem as an approximation of surfaces, see \refsec{surfacerendering}.
    \item \textit{Depth-sorting of intersections:}
          Once all intersections are determined, they are sorted along each view-ray according to the distance from the origin (depth), see \refsec{depthSorting}.
    \item \textit{Approximate Volume Render Integral:}
          For each pair of entry and exit points the volume render integral~\eqref{eq:volRenderIntegral} is approximated, using \textit{pixel-accurate} approximations of the \textit{inverse of the view-rays}, see \refsec{geometry}.
\end{enumerate}

\subsection{View-Ray Intersections with Surfaces}\label{sec:surfacerendering}

Computing the intersections between a ray and a spline surface can be a computationally expensive and unstable operation.
Commonly, the problem is stated as finding the zeros of a function.
However, the intersection problem can be restated as the problem of finding a view-dependent approximation of the surfaces of the object.
Therefore, the rasterization process of GPUs can be used as a very efficient, parallel implementation of finding all ray-surface intersections of a triangulation.
Two alternative approaches to construct a view dependent triangulation using the hardware tessellator where recently presented by Yeo et al.~\cite{Yeo:2012} and by Hjelmervik~\cite{hjelmervik2012direct}.
Both methods are applicable in our setting, guaranteeing both water tightness and that the approximation error satisfies the requirements from \refdef{pixelaccuracy}.
Our implementation uses \cite{hjelmervik2012direct} since it provides a single-pass algorithm.

In our approach all boundary surfaces of all volume blocks $G_\alpha$ are rendered (6 surfaces per block) and the result is stored in a texture buffer.

\subsection{Blockwise Depth-Sorting of Intersections}\label{sec:depthSorting}

The first stage of our approach (see \refsec{surfacerendering}) will lead to an unordered list of intersections per view-ray.
The number of intersections depends on view-angle, and the following two properties, depicted in \reffig{depthsorting}:
\begin{itemize}
    \item \textit{Non-convex objects:}
        For non-convex objects, the view-ray can intersect the geometry multiple times per block, leading to multiple entry and exit points along the view-ray.
    \item \textit{Multi-block objects:}
        Each volume block potentially adds further intersections.
\end{itemize}

Since modern GPUs allow atomic operations, our approach is based on a linked list (per pixel-location), which is populated in a single pass, as detailed in Yang et al.~\cite{Yang:2010:RCL:2383616.2383624}.
This method comes from a problem known as \textit{order-independent transparency} in computer graphics, see e.g., Everitt~\cite{Everitt01ii}, Carpenter and Ii~\cite{Carpenter84thea-buffer}, and Myers and Bavoil~\cite{Myers:2007:SRA}.

This list is the basis for volume rendering in the next stage.
In addition to the parameter value, each entry stores information about the depth, the block number, and the orientation (front facing or not).
In order to produce a pixel-accurate result according to \refdef{pixelaccuracy}, we choose pairs from the list of intersections according to \refalg{depthsorting}.
\begin{algorithm}[t]
    \caption{Blockwise depth-sorting}
    \label{alg:depthsorting}
    \begin{enumerate}
        \item Find $\pback$ by going through the list and choose entry
            \begin{itemize}
                \item with depth value closest to screen,
                \item that is not marked as used,
                \item and is not front facing
            \end{itemize}
        \item exit if the list is empty
        \item Find $\pfront$ by going through the list and choose entry
            \begin{itemize}
                \item with depth value closest to screen,
                \item that is not marked as used,
                \item that is front facing,
                \item and matches block number of $\pback$
            \end{itemize}
        \item render volume between $\pfront$ and $\pback$
        \item mark $\pback$ and $\pfront$ as used
    \end{enumerate}
\end{algorithm}

In IGA the surfaces of two neighboring blocks match exactly.
The pixel-accurate rendering of the surfaces described in \refsec{surfacerendering} will produce two entries in the list with approximately the same depth value.
However, \refalg{depthsorting} ensures that the found pairs correspond to the same block number.

\section{Pixel-Accurate Inverse of View-Rays}\label{sec:geometry}

In the previous \refsec{approach} we described how to obtain the intersections of the view-ray $\gamma$ with the surfaces of the volume in the parameter domain $P$, providing pairs of entry and exit points ($\gfront$ and $\gback$) for the compositing scheme.
In order to approximate the volume render integral \eqref{eq:volRenderIntegral} pixel-accurately (see \refdef{pixelaccuracy}), approximations for sample points between $\gfront$ and $\gback$ have to be found.

This section describes two alternative methods for finding an approximation of the inverse of the function $\phi$ for all points on the view-ray between $\gfront$ and $\gback$.
If not otherwise mentioned, we will assume that each pair of entry and exit points
\begin{enumerate}
    \item is given and pixel-accurate (ensured by \refsec{surfacerendering}),
    \item belongs to only one volume-block, denoted by $P_\alpha$ (ensured by \refsec{depthSorting}),
    \item the line connecting the entry with the exit point does not intersect the boundary of block $P_\alpha$ (this case is described in \refsec{cornercases}), and
    \item $\phi_\alpha$ is a diffeomorphism with a continuous first derivative on the block $P_\alpha$ (see \refsec{cornercases} for the case when $\Jac{\phi}$ is singular on the boundary $\partial P_\alpha$).
\end{enumerate}

We would like to point out that, under assumptions 1.-4., the problem of finding the inverse  is well-posed in the sense of Hadamard, since $\phi_\alpha$ is continuous, bijective and differentiable within the interior of each block.
The aforementioned conditions on $\phi_\alpha$ are reflected in the requirements for the numerical analysis performed on the isogeometric object.

The two alternative methods described in the following have certain advantages and disadvantages, discussed below.
It is worth noting that neither method is limited to the case where $\phi$ is a trivariate spline, but works as long as above assumptions on $\phi$ hold.

\begin{figure}
    \centering
    \includegraphics[width=0.99\linewidth]{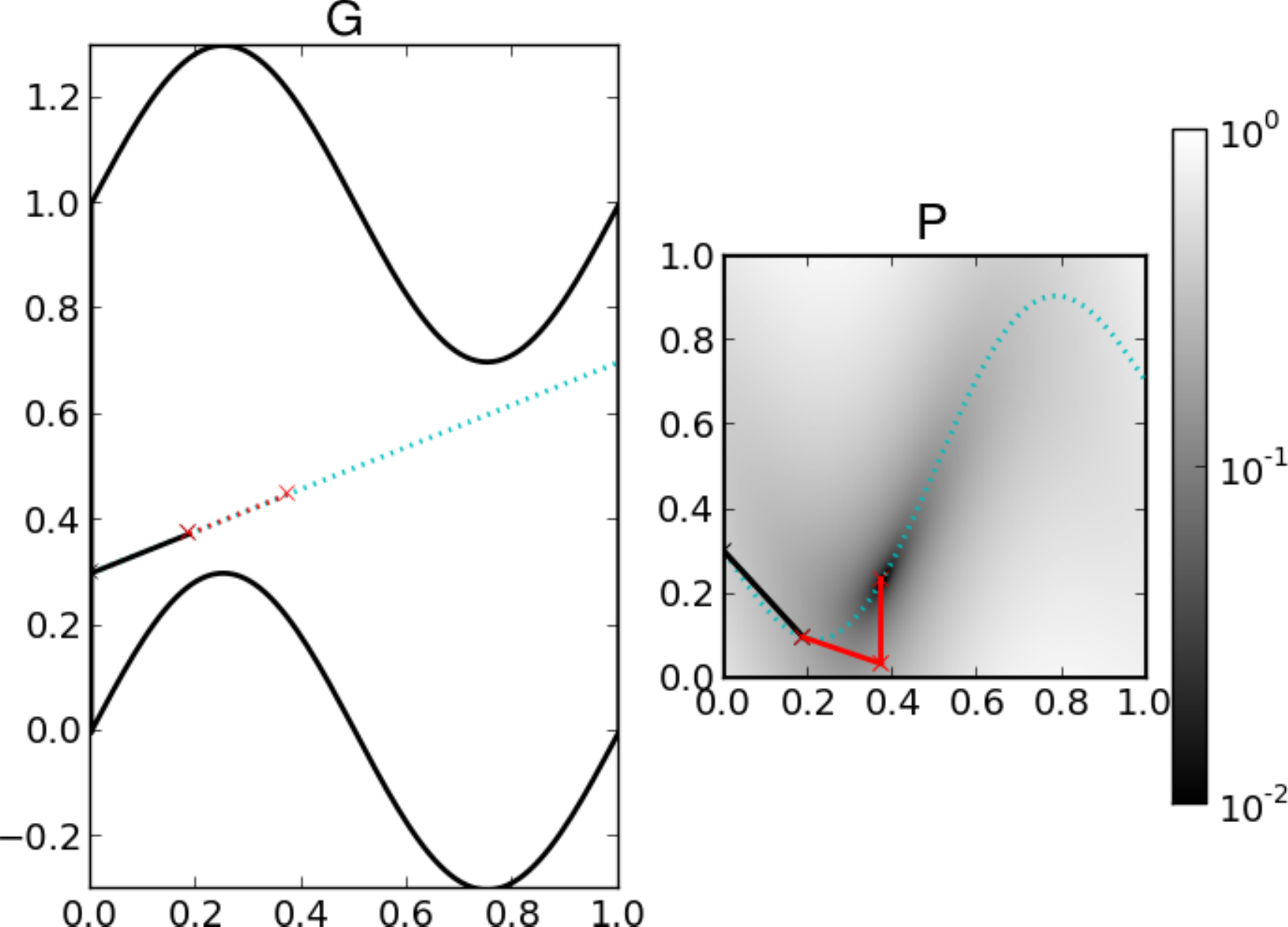}
    \caption{
        One step in the root finding based method for a two dimensional example, where $\phi(x,y)=(x,0.3\sin(2\pi x))$.
        The red lines in the parameter domain P show how the Newton-Raphson method iteratively converges towards the zero of \refequ{F}.
        The underlain image in P shows the norm of $F$.
    }
    \protect \label{fig:2dnewton}
\end{figure}

\subsection{Root Finding Based Algorithm.}\label{sec:root}

This first approach has been outlined by Martin and Cohen~\cite{Martin_Cohen:2001:REVAUTS}, but no practical details are discussed and, to the best of our knowledge, no implementation exists.
We will provide a brief description of the algorithm, followed by a discussion of the approach.

Given $\phi : P \rightarrow G$ the problem is to find the preimage of sample points $g_i = g_{i-1} + \Delta s_i \frac{\gback-\gfront}{||\gback-\gfront||_{L^2}}$ along the view ray.
Here, $g_0 = \gfront$ and $\Delta s_i$ is a (variable) sample distance.
Mathematically, the point $p_i$ that is the preimage of $g_i$ is in the null space of the following function
\begin{equation}\label{eq:F}
    F_{g_i}: P\rightarrow\R^3: p \mapsto \phi(p)-g_i.
\end{equation}
A standard method for finding approximations of the roots of function \eqref{eq:F} is the Newton-Raphson method,
given by
\begin{align}\label{eq:newtonsmethod}
    \Jac{F}(x_n) (x_{n+1}-x_n) &= -F(x_n),
\end{align}
where $\Jac{F}$ denotes the Jacobian matrix of $F$, and $x_n$ are the approximations to the root of $F$.
To solve the $3\times3$ linear system of equations \eqref{eq:newtonsmethod} we employ the QR-algorithm, see e.g., \cite{deuflhard2008numerische}.
See \reffig{2dnewton} for an example.

Note that for each iteration of \eqref{eq:newtonsmethod}, both $\phi$ and $\Jac{F} = \Jac{\phi}$ have to be evaluated at the same point.
Since $\phi$ is a spline function the Jacobian can be evaluated cheaply by reusing calculations for $\phi$.

The Newton-Raphson method converges quadratically for "good" starting points. The method can, however, fail in certain situations.
We will address each situation in the following for the problem at hand.

For a good argumentation let us note that a bijective $\phi$ induces a metric space $(P,d_P)$ on the open parameter domain $P$ of each block with
\begin{equation}\label{eq:metricOnP}
    d_P(p_1,p_2) \coloneqq \|\phi(p_1)-\phi(p_2)\|_{L^2}, \quad p_1, p_2 \in P.
\end{equation}
According to our basic assumptions, \refequ{F} has a unique solution and therefore a root of $F$ is a root of the norm of $F$ and vice versa.
Since
\begin{equation}\label{eq:equivalencetodist}
    \|F_{g_i}(p)\|_{L^2}=d_P(p,\phi^{-1}(g_i)),
\end{equation}
the problem of finding the inverse is equivalent to finding the minimum distance to the preimage of $g_i$.
As a consequence, $F_{g_i}$ will  not have horizontal asymptotes or local extrema/stationary points, and converge for all starting points in the interior of the block.
Overshooting can be an issue in case of large distance between samples or where the geometry is complicated.
Since we have a bounded domain $P$ we need to clamp values to the range of the parameter domain, usually $[0,1]^3$,
and use the method of line search.
We would also like to point out that since $\Jac{F} = \Jac{\phi}$ the Jacobian in \refequ{newtonsmethod} is non-singular on $P\backslash \partial P$ due to our basic assumptions.

The Newton-Raphson method \eqref{eq:newtonsmethod} needs a stopping criterion.
In our case we require the method to be pixel-accurate, so we stop the iteration when the distance to the view-ray (given by $\|F_{g_i}(x_n)\|_{L^2}$) is less or equal to the minimum distance of the point $g_i$ to the frustum boundary.

\subsection{ODE Based Algorithm.}\label{sec:ode}

We will now describe a second approach based on the fact that the view-ray $\gamma$ can be seen as an integral curve of a (first order linear) dynamical system, i.e., the solution of an ordinary differential equation (ODE).
The idea is to directly work in the parameter domain by appropriately defining an ODE for which the image of the solution coincides with the original view-ray $\gamma$.
This can be achieved by defining a vector field and numerically approximating the integral curve from the point where the ray enters the domain, see \reffig{influenceofC}.

\begin{figure}
\centering
    \centering
    \subfigure[Vector field \refequ{vec2} with no perpendicular component, i.e., $c=0$.]{\includegraphics[width=0.22\textwidth]{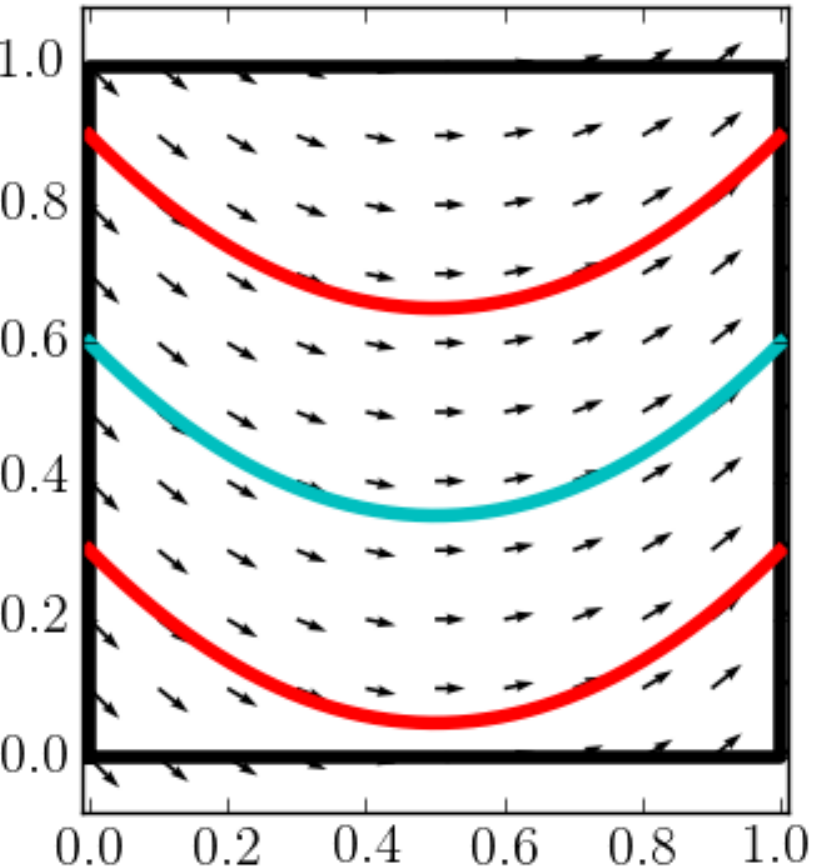}}
    \subfigure[Vector field \refequ{vec2} with perpendicular component $c=20$.]{\includegraphics[width=0.22\textwidth]{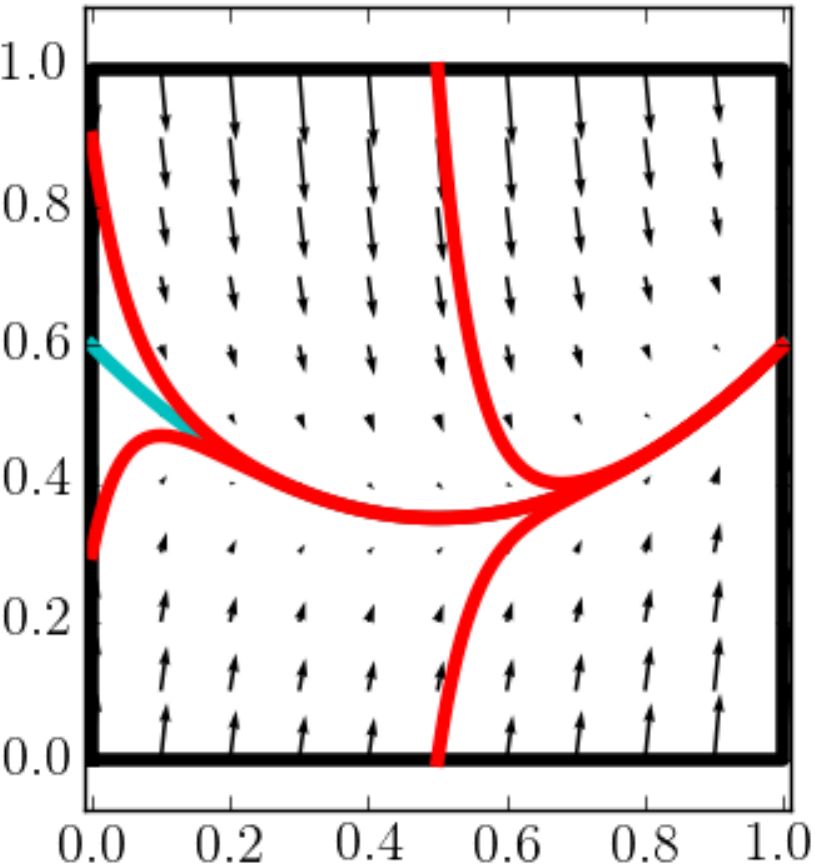}}
    \caption{
        The view-ray is shown in cyan, and streamlines for different start points are shown in red.
        The perpendicular component makes sure that the solution quickly converges towards the view ray.
    }
    \protect \label{fig:influenceofC}
\end{figure}

In order to describe our method we start by defining an ODE on the geometry domain through
\begin{equation}\label{eq:odegeo}
        g'(s) = V(g(s)), \quad \quad g(0) = \gfront.
\end{equation}
We define the vector field to be
\begin{equation}\label{eq:vec2}
    V(g)= \Vpar + c V_\perp(g),
\end{equation}
consisting of two components, where the constant $c$ is the (positive) relative weight between the two components.
The first component is a constant velocity parallel to the view-ray given by $\Vpar = \frac{\gback-\gfront}{||\gback-\gfront||}$.
The second component is a velocity perpendicular to the view-ray and depends on the signed distance to the view-ray, i.e.,
\begin{equation}
    V_\perp(g) = (\gfront - g) - <\gfront-g,\Vpar>\Vpar,
\end{equation}
where $<.,.>$ is the usual scalar product.
The motivation for introducing the perpendicular component $V_\perp$ is the following. 
The numerical method will inevitably introduce errors.
The larger $c$ the more will the numerical approximation be forced back to the exact solution, see \reffig{influenceofC} for an illustration.

The ODE \eqref{eq:odegeo} is a first order linear dynamical system, which can be rewritten in the standard form
\begin{align}\label{eq:odegeoMatrix}
    \begin{split}
        g' &= A g + b, \ \text{with} \\
        A &= c\begin{pmatrix}
              \Vpara^2 - 1 &  \Vpara\Vparb & \Vpara\Vparc \\
              \Vparb\Vpara & \Vparb^2 - 1  & \Vparb\Vparc \\
              \Vparc\Vpara & \Vparc\Vparb  & \Vparc^2 - 1
             \end{pmatrix},\\
             b &= c\gfront + (1 - c<\gfront,\Vpar>)\Vpar.
    \end{split}
\end{align}
The dynamics are determined by the eigenstructure of $A$, see e.g., \cite{perko2001differential}.
There are two negative eigenvalues $\lambda_{1,2} = -c$ and one zero eigenvalue $\lambda_{3} = 0$ with corresponding eigenvector
$v_3 = \Vpar$.
We can write a vector $u\in\R^3$ as the sum $u=u_\perp + u_\parallel$, where $\R^3 = M_\perp \oplus M_\parallel = \text{span}(v_1,v_2) \oplus \text{span}(v_3)$.
Furthermore, there exists a unique vector $m \in M_\perp$ such that $A m = b_\perp$.
The solution of \refequ{odegeoMatrix} is then
\begin{equation}
    g(t) = e^{At}(g_\parallel+m) - m + t b_\parallel,
\end{equation}
where $g_\parallel, b_\parallel \in M_\parallel$.
This analysis shows that the view-ray is a solution and that all solutions (irrespective of the starting point $g(0)$) converge exponentially towards the view-ray.

\begin{figure}
    \centering
    \begin{tikzpicture}[
        scale=0.6,
        grid/.style={very thin,gray},
        cube/.style={thick,line join=round},
        cube hidden/.style={thick,dashed,line join=round}]

        \def\translatePAR{6}
        \def\translateGEO{0}
        \node (P) at (\translatePAR+1,2.5) {P};

        \node (G) at (\translateGEO+1.5,2.5) {G};
        \node [anchor=south](geo) at ($(G.south)+(0,-7.5)$) {\includegraphics[width=0.38\linewidth]{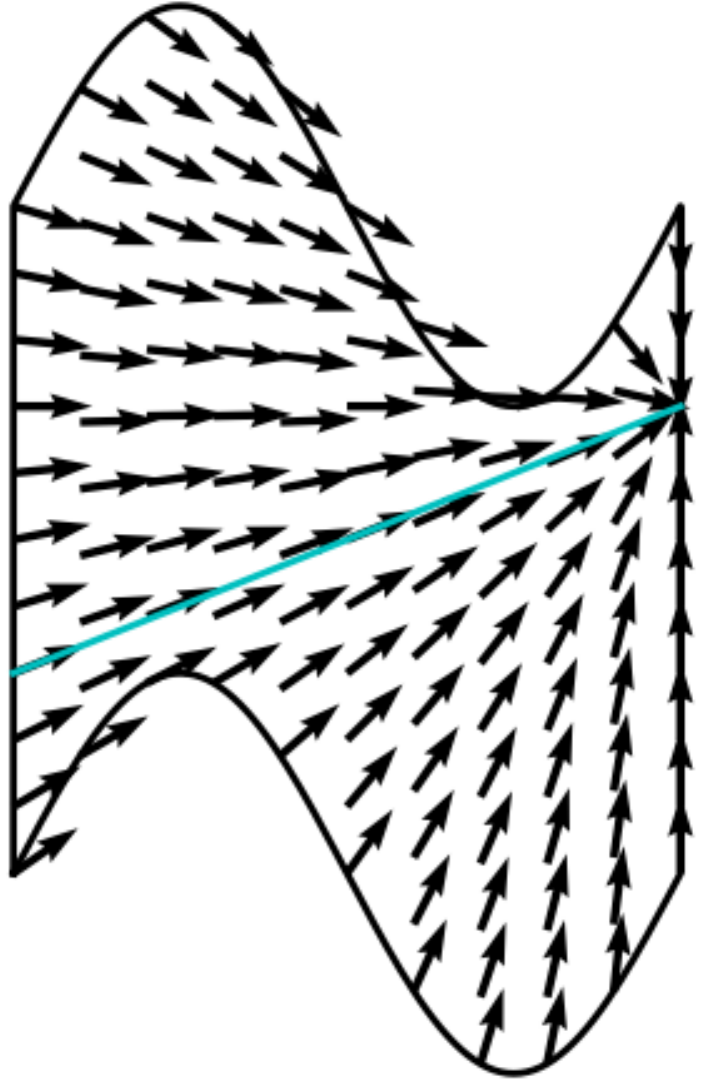}};
        \node [anchor=south](par) at ($(P.south)+(0,-6)$) {\includegraphics[width=0.38\linewidth]{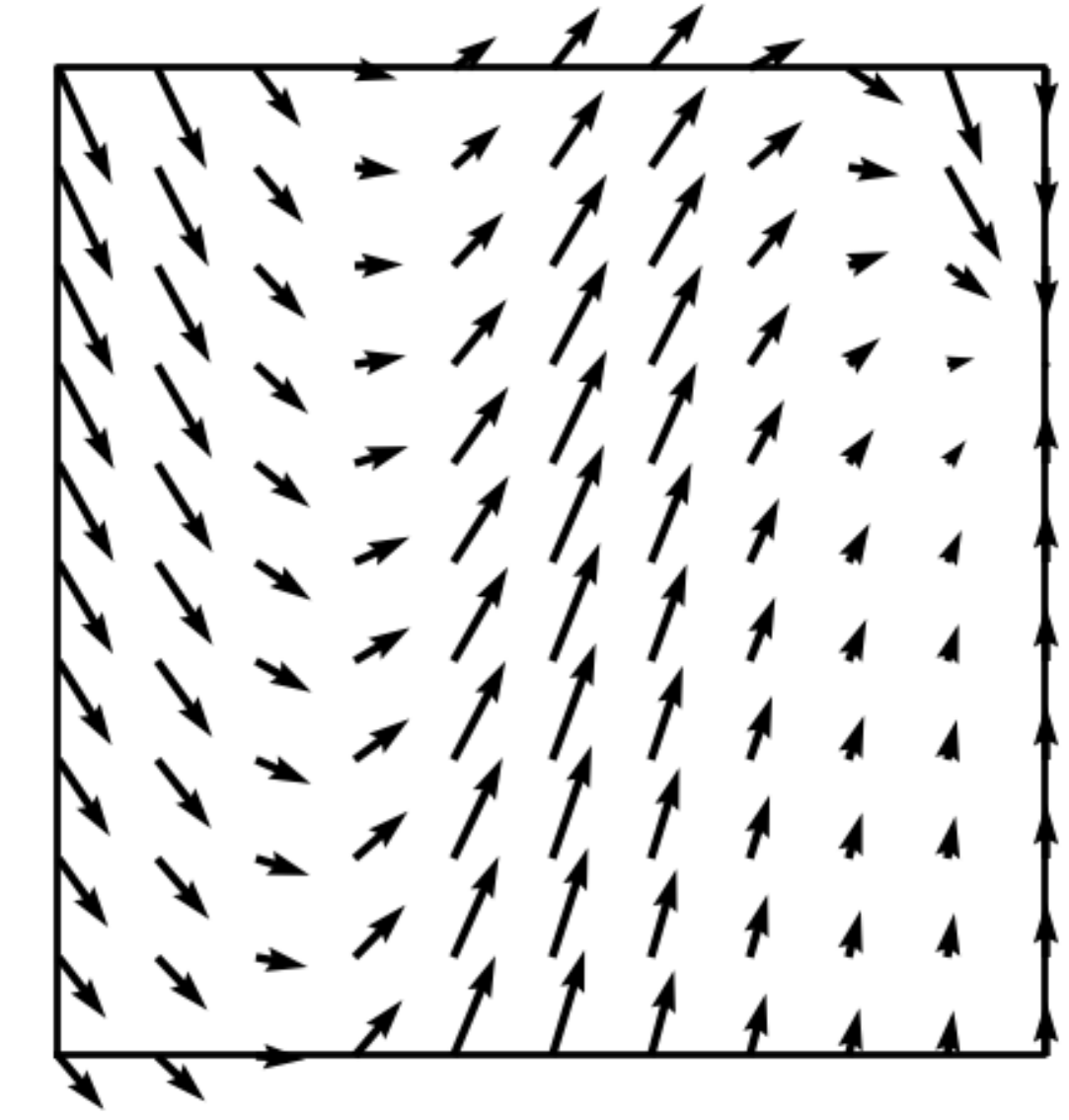}};
        \node (G) at (\translateGEO+1.5,2.5) {G};
        \draw[thick,->] (P) to node[above] {$\phi$} (G);
        \node (TP) at (\translatePAR+1,4.5) {TP};
        \node (TG) at (\translateGEO+1.5,4.5) {TG};
        \draw[thick,->] (TP) to node[above] {$\phi_*$} (TG);
        \draw[thick,->] (TP) to node[right] {$\pi$} (P);
        \draw[thick,->] (TG) to node[left] {$\pi$} (G);
    \end{tikzpicture}
    \caption{
        Tangent bundle for $\phi(x,y)=(x,y+0.3\sin(2\pi x))$ with vector field defined on G and pulled back to P.
        Projection given by $\pi(\omega,v) = \omega$.
    }
    \protect \label{fig:tangentbundlediagram}
\end{figure}

Next we will define a vector field $W(p)$ such that the solution of the ODE on the parameter domain P
\begin{align}\label{eq:odepar}
    \begin{split}
        p'(s) &= W(p(s)), \quad (s,p) \in \R \times [0,1]^3, \\
        \quad p(0) &= \pfront,
    \end{split}
\end{align}
coincides with the preimage of the solution of \refequ{odegeo}.

Using that $\phi$ is a differentiable map on the inside of $P$, the dual is given by
\begin{equation}
    \phi_{*}: TP\rightarrow TG, (p,v) \mapsto (\phi(p), \Jac{\phi}(p) v),
\end{equation}
where $T$ denotes the tangent space, see e.g., Spivak~\cite{spivak1979comprehensive}.
We then have that the diagram shown in \reffig{tangentbundlediagram} commutes.

Since $\phi$ is a diffeomorphism on the inside of $P$, it follows that the dual has an inverse given by,
\begin{equation}\label{eq:invdual}
    \phi_{*}^{-1}: (p,v) \mapsto (\phi^{-1}(p), (\Jac{\phi}(p))^{-1} v),
\end{equation}
using the inverse function theorem.
This can be used to "pull-back" the vector field $W(p)$ in \refequ{odepar}, by solving the following system of linear equations:
\begin{equation}\label{eq:VgVp}
    \Jac{\phi}(p) W(p) =  V(\phi(p)),
\end{equation}
Note that $\Jac{\phi}$ exists and is non-singular for all $p$ in the inside of $P$.
See \reffig{2dcomparisonvectorfield} for an example of how the vector field becomes non-trivial in the parameter domain.

Since the vector field $W$ is in general non-linear, we establish the following theorem.

\begin{theorem}[Existence and Uniqueness]\label{theorem:odeExUn}
    Under the condition that $\phi$ is a diffeomorphism with a continuous Jacobian, there exists a unique solution to the initial value problem \eqref{eq:odepar} (with \eqref{eq:VgVp}) that continues up to the boundary.
\end{theorem}
\begin{proof}
    According to Theorem 2.7 in \cite{deuflhard2002scientific} (based on the theorem of Picard-Lindel\"{o}f) it is enough to show that the right hand side of \refequ{odepar} is continuous and locally Lipschitz-continuous.
    For every compact subset of $S \subset P$ there exists a constant $C_\phi$, such that for all $v_1,v_2\in\R^3$ and $p_1,p_2\in S$
    \begin{equation*}
        \| \Jac{\phi}^{-1}(p_1)v_1 - \Jac{\phi}^{-1}(p_2)v_2\| \leq C_\phi \| v_1 - v_2\|,
    \end{equation*}
    since $\phi$ has a continuous Jacobian.
    It is then easy to show that 
    \begin{equation}\label{eq:PicLind}
           \| W(p_1)-W(p_2)\| 
           \leq C_\phi \|A\|_{op} L_{\phi} \| p_1 - p_2\|,
    \end{equation}
    where $\|.\|_{op}$ is the operator norm and $L_{\phi}$ is the Lipschitz constant of $\phi$.
\end{proof}

Solutions of the ODE \eqref{eq:odepar} can be approximated using a wide variety of numerical methods, e.g., explicit Runge-Kutta methods (see \cite{deuflhard2002scientific}).
The simplest scheme is the explicit Euler method given by
\begin{equation}\label{eq:explicitEuler}
    p(s^{n+1}) = p(s^n) + \Delta s^n W\left(p(s^n)\right),
\end{equation}
where we use standard notation using superscripts indicating discrete "times" and $\Delta s^n = (s^{n+1}-s^n)$. 
Note that in each step we have to solve the linear system \eqref{eq:VgVp}.

The stiffness index of \refequ{odegeoMatrix} is $L=\max_i(|\lambda_i|) = c$, see e.g., \cite{deuflhard2002scientific}. While larger $c$ have the benefit of preventing the numerical approximation to deviate to far from the view-ray (see \reffig{2dcomparisonvectorfield}), the ODE will become increasingly stiff. This behavior will be inherited in the ODE on P as well.
Explicit solvers, such as the one described in \refequ{explicitEuler}, will suffer from impractically small $\Delta s^n$ \emph{for large $c$}.
Generally, implicit schemes have larger stability regions.
An \emph{A-stable} method suitable for stiff equations is the \emph{implicit Euler} scheme, given by
\begin{equation}\label{eq:implicitEuler}
    p(s^{n+1}) = p(s^n) + \Delta s^n W(p(s^{n+1})),
\end{equation}
where a system of non-linear equations must be solved in each step.
In order to achieve maximum stability, we do not solve for $p(s^{n+1})$ directly, but rather for the difference to the previous point, see e.g., ~\cite{deuflhard2002scientific}.
Thus, we have to solve
\begin{equation}\label{eq:newtonimpliciteuler}
    G(z) = z - \Delta s^n W(z + p(s^{n})) = 0,
\end{equation}
numerically, for which we apply the Newton-Raphson method in order to approximate the solution.
This method depends on the Jacobian of $G$, which can be approximated.
However, in order to increase performance we avoid further evaluations of $G$ and instead derive an exact expression for $\Jac{G}$.
By rewriting \refequ{newtonimpliciteuler} as
\begin{equation}\label{eq:newtonimpliciteulerReformulated}
    \Jac{\phi}(z+p^n)  G(z) = \Jac{\phi}(z+p^n) z - \Delta s^n V(\phi(z + p^n)),
\end{equation}
where $p^n=p(s^{n})$ and applying the Jacobian operator to both sides, we derive (using the chain- and product rule) the following linear system of equations
\begin{align}\label{eq:JacG}
    \begin{split}
        \Jac{\phi}(z+p^n) \Jac{G}(z) = H_{\phi}(z+p^n) (z - G(z)) \\
         + \left(\mathcal{I}- \Delta s \Jac{V}(\phi(z + p^n)\right) \Jac{\phi}(z+p^n).
    \end{split}
\end{align}
Here, $\mathcal{I}$ is the identity matrix and $H_\phi$ denotes the Hessian of $\phi$, a tensor of order 3.
Observe, that the spline $\phi$ along with its first and second order derivatives are evaluated at the same point, allowing for an efficient calculation in a shader program.
The Jacobian of the vector field $V$ is the following constant matrix
\begin{equation}
    \Jac{V}(g) = c\left( C -\mathcal{I} \right),
\end{equation}
where the i-th column of C is given by $V_{\parallel,i} V_\parallel$.

By solving the linear \refequ{JacG} we can calculate the Jacobian matrix of $G$ and use it for the Newton-Raphson method used for the implicit Euler method. It has the same matrix as the equation we need to solve to get the vector field $W$, see \refequ{VgVp}. This means that, when using the QR-algorithm for solving both linear equations, an efficient algorithm can reuse Q and R for solving \refequ{JacG}.

\subsection{Convergence study.}\label{sec:convergencestudy}
\reffig{2dcomparisonvectorfield} shows a two dimensional test case given by 
$\phi(x,y)=(2x,y+0.3(1-x)\sin(10\pi x))$ with $\pfront = (0,0.3), \pback = (1,0.7)$.
In \reftable{2dcomparisonefficiency} we show how the different numerical methods converge to the exact solution.
We use the following notation:
(RK 1) explicit Euler method (first order),
(IRK 1) implicit Euler method (first order),
(RK 2) midpoint method (second order),
(RK3) Kutta's 3rd order method,
(RK4) classic 4th order method,
(RK4 3/8) 3/8 rule (4th order),
(RKF) Runge-Kutta-Fehlberg method (5th order),
and (RF) root finding method.
We observe that the ODE based algorithms converge with the expected order as the sample distance is reduced and the root finding based method reaches the given tolerance for all sample distances.

\begin{figure}
    \centering
    \includegraphics[width=0.90\linewidth]{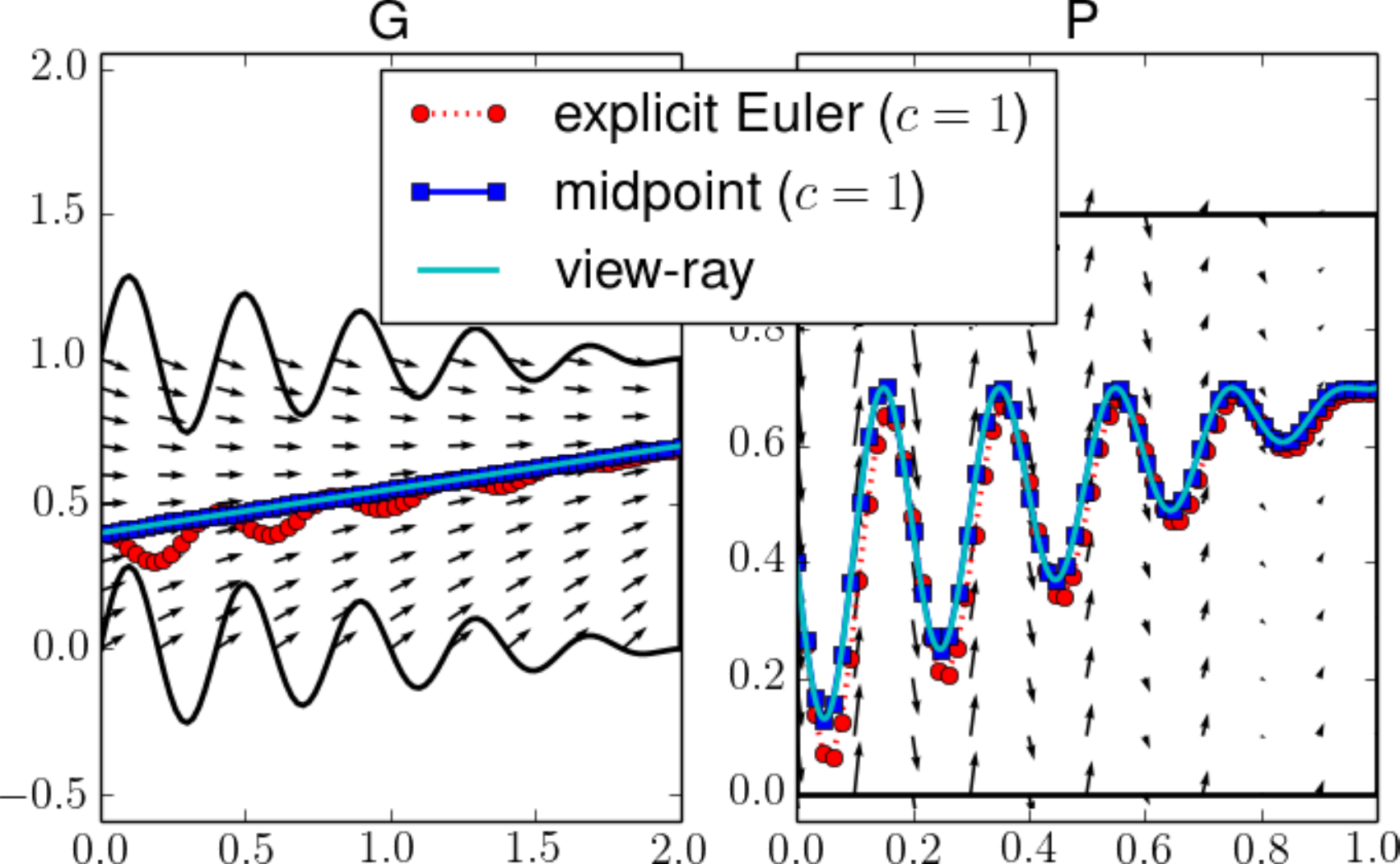}
    \caption{
        With the ODE based method, the preimage of the view-ray in the geometry domain $G$ is given by an integral curve in the parameter domain P. 
    }
    \protect \label{fig:2dcomparisonvectorfield}
\end{figure}
\begin{table}[t]
    \centering
    \begin{tabular}{| l | c | c | c | c |}
        \hline
        \multicolumn{1}{|c|}{$\Delta s$} & 1.6e-02 & 7.8e-03 & 3.9e-03 & 2.0e-03 \\
        \hline
        RK 1 (c=1) &
        6.2e-02 & 3.1e-02 & 1.6e-02 & 7.8e-03 \\
        IRK 1 (c=100) &
        5.3e-03 & 1.6e-03 & 8.1e-04 & 5.4e-04 \\
        RK 2 (c=1) &
        8.6e-04 & 2.1e-04 & 5.2e-05 & 1.3e-05 \\
        RK 3 (c=1) &
        1.5e-06 & 1.7e-07 & 2.1e-08 & 2.5e-09 \\
        RK 4 (c=1) &
        4.2e-07 & 3.0e-08 & 2.0e-09 & 1.3e-10 \\
        RK 4 3/8 (c=1) &
        1.7e-07 & 1.3e-08 & 8.8e-10 & 5.7e-11 \\
        RKF (c=1) &
        3.0e-08 & 9.5e-10 & 2.9e-11 & 9.2e-13 \\
        RF (tol = 1e-3) &
        9.8e-04 & 9.8e-04 & 5.2e-04 & 1.3e-04 \\
        RF (tol = 1e-14) &
        2.5e-16 & 2.5e-16 & 2.5e-16 & 2.5e-16 \\
        \hline
    \end{tabular}

    \caption{
        For the case depicted in \reffig{2dcomparisonvectorfield}, the ODE based algorithms (for notation see \refsec{convergencestudy}) show the expected convergence rates, and the root finding based method (RF) reaches the given tolerance.
        The error is defined by $e_{L^\infty} = \max_i \| (\pfront - \phi(p_i) - <\pfront-\phi(p_i),V_\parallel>V_\parallel \|_{L^2}$.
    }
    \protect \label{table:2dcomparisonefficiency}
\end{table}

\subsection{Degeneracies and Points Outside Domain.}\label{sec:cornercases}

There are two prominent cases for which the methods described in sections~\ref{sec:root} and \ref{sec:ode} need minor adjustments.
The first case is when the line between $\gfront$ and $\gback$ intersects the boundary of the volume block.
Although this is a rare case, it can happen that the approximation of the surfaces "misses" intersections in the tessellation of the geometry (described in \refsec{surfacerendering}).
This case is shown in \reffig{intersecting}, where the approximated surface (dashed black line in geometry domain) is still pixel-accurate, but the view-ray intersects the exact surface.
In such a case, the Newton-Raphson method in both the implicit Euler method as well as the root finding based method will not converge, but repeatedly try to exit the parameter domain of the corresponding block.
We detect such behavior and step along the boundary of the parameter domain until the view-ray is within the domain again.
For the explicit Runge-Kutta methods we simply clamp the approximated solution values to remain in $P$.
Observe, that the resulting approximation of the view-ray (seen in \reffig{intersecting}) is still pixel-accurate, since the approximation of the surface is guaranteed to be so.

The second case is when there are degeneracies of the spline $\phi$ along the boundary, see e.g., \reffig{DemonstratorQuality}.
Assume for instance that the Jacobian $\Jac{\phi}$ is singular at the entry point $\gfront$.
Since both the root finding method (\refsec{root}) as well as the ODE based methods (\refsec{ode}) involve solving a system of linear equations with a singular matrix in that case (see Equations~\eqref{eq:newtonsmethod} and \eqref{eq:VgVp}), the only viable choice is to "shrink" the block in the following way.
By choosing a new entry point $\tildegfront = \gfront + \delta \frac{\pback-\pfront}{||\pback-\pfront||_{L^2}}$, with a suitable $\delta>0$.
The new entry point $\tildepfront$ in the parameter domain $P$ can be found with the root finding method with a different starting point for the iterations in the Newton-Raphson method, for instance
\begin{equation}\label{eq:degeneracyfix}
    x_0 = \pfront + \eps \frac{\pback-\pfront}{||\pback-\pfront||_{L^2}},
\end{equation}
with an appropriately chosen $\eps$.
As can be seen in the middle of \reffig{teaser} this approach works well.

\begin{figure}
    \centering
    \includegraphics[width=0.9\linewidth]{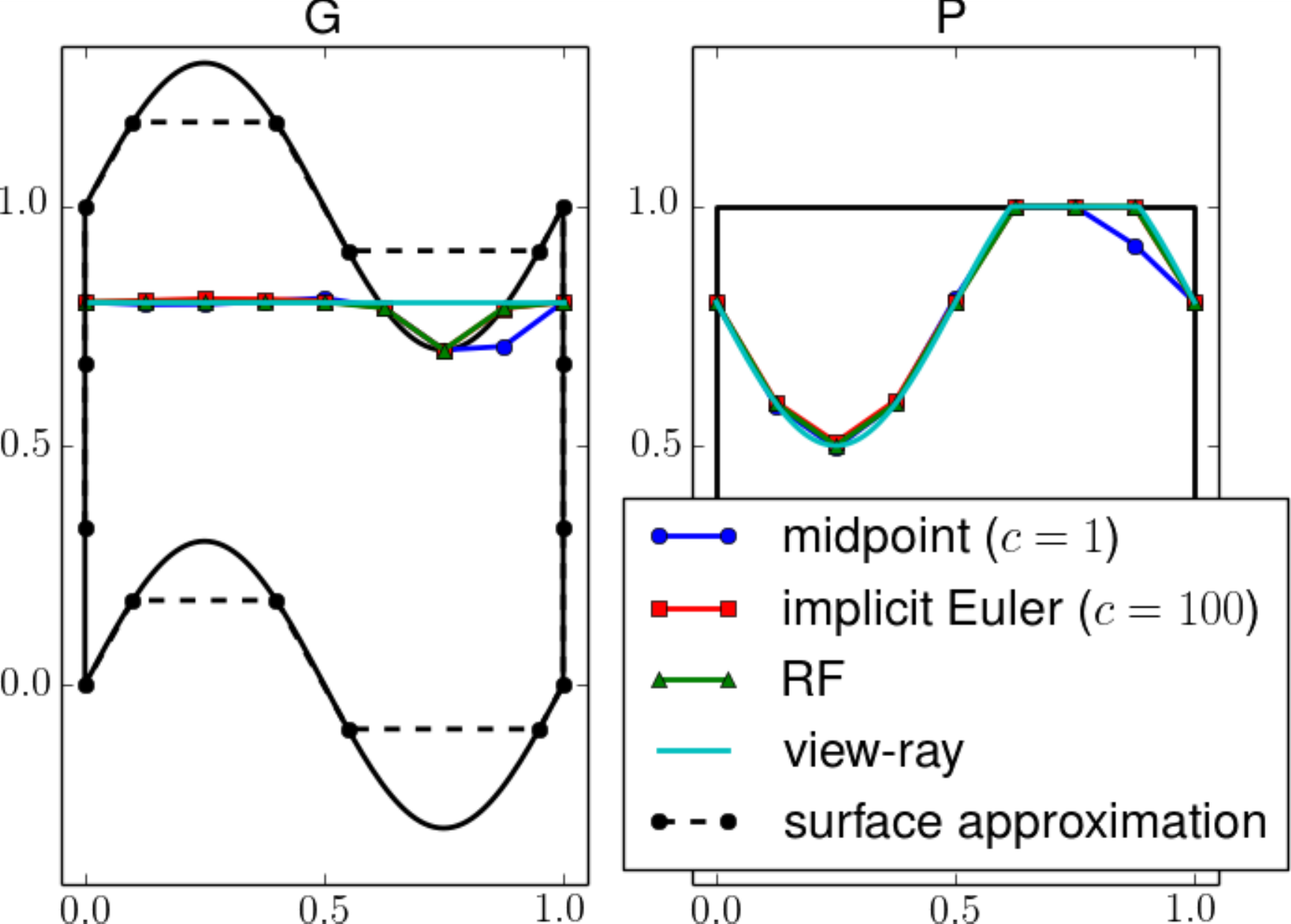}
    \caption{When the view-ray intersects the boundary of a block due to a large pixel-frustum, both the root finding method (RF) and the methods based on ODEs still lead to a pixel-accurate approximation.}
    \protect \label{fig:intersecting}
\end{figure}

\subsection{Cutting and Near Clip Planes}\label{sec:cuttingplane}

The root finding based method is well suited for realizing cutting planes.
Before the final compositing, the depth-sorted list of intersections (from \refsec{depthSorting}) is handled in the following way
\begin{itemize}
    \item \textit{Cutting planes}:
          The plane is given by a point $g_0\in G$ and a normal $n$.
          If $(n\cdot\gfront-g_0)$ and $(n\cdot\gback-g_0)$ have opposite signs, the view-ray $\gamma$ between $\gfront$ and $\gback$ intersects the cutting plane.
    \item \textit{Near-Clip plane}:
          Since the rendering of the surfaces is water tight, an odd number of ray-surface intersections means that the near plane is inside the volume.
\end{itemize}
In both cases, trivial formulas determine the intersection point $g_* \in G$.
Given $g_*$, we need to find the corresponding point $p_*\in P$ such that $\phi(p_*) - g_* = 0$.
Thus, finding $p_*$ is exactly solving \refequ{F}.
Since this point can be quite far away from $\pfront$ or $\pback$ the root finding algorithm described in \refsec{root} is the best alternative and can readily be used.

\section{Applications and Performance}\label{sec:applications}

In order to benchmark the performance of the proposed methods, we present the three different application scenarios shown in \reffig{teaser}, covering a wide range of possible applications. 
We apply the approach described in Sections~\ref{sec:approach} and \ref{sec:geometry} in each case, and compare it with standard volume rendering algorithms, where we have precomputed a voxelized version of the model.
The resulting texture has 16 bit and uses the red channel for the scalar value and the green channel to encode if the voxel is inside the object or not.
Of course, many optimization strategies are established in standard volume rendering, such as adaptive sampling rates, out of core algorithms, et cetera.
However, to allow for a fair comparison, we only use an out of the box implementation without any optimizations.
In all cases we measure the performance of our algorithm on an NVIDIA Titan GPU.

The proposed approach, allocates memory for the knots and the control points.
In addition, a buffer is allocated for the linked list containing all view-ray intersections with the surfaces (see \refsec{depthSorting}).

In order to measure how well the volume render integral~\eqref{eq:volRenderIntegral} is approximated we measure the color difference to a reference image with the norm "$\Delta E$" as defined by the International Commission on Illumination (CIE) in 2000.
In this metric $\Delta E=1$ means a "just noticeable difference".
We will, however, be less strict and regard values of up to 5 to be acceptable.

We start with an application useful in the design phase of the geometry.

\setlength{\tabcolsep}{.12em}
\begin{table*}
    \centering
    \subtable[ODE based method RK4 3/8 (fourth order).]{
        \begin{tabular}{| r | r | r | r | r | r | r | r |}
            \hline
            max(\# S) & max($\Delta P$) & max($\Delta E$) & mean($\Delta E$) & var($\Delta E$) & surf [ms] & ray [ms] & tot [ms] \\
            \hline
            11  &  2.1  & 55.064  & 2.946  & 21.834  & 0.86 & 8.02 & 9.15 \\
            23  & \cellcolor{lightgray} 0.6  & 34.447  & 1.421  & 7.211   & 0.86 & 13.74 & 14.91 \\
            47  & \cellcolor{lightgray} 0.6  & 17.771  & 0.537  & 1.092  & 0.86 & 24.92 & 26.18 \\
            95  & \cellcolor{lightgray} 0.6  &\cellcolor{lightgray} 4.728  & 0.189  & 0.111    & 0.86 & 46.67 & 47.97 \\
            190  & \cellcolor{lightgray} 0.6  &\cellcolor{lightgray} 2.105  & 0.077  & 0.022   & 0.86 & 88.22 & 89.74 \\
            381  & \cellcolor{lightgray} 0.6  &\cellcolor{lightgray} 1.167  & 0.038  & 0.012  & 0.86 & 172.70 & 174.50 \\
            762  & \cellcolor{lightgray} 0.6  &\cellcolor{lightgray} 1.092  & 0.020  & 0.006  & 0.86 & 337.30 & 339.70 \\
            1525  & \cellcolor{lightgray} 0.6  &\cellcolor{lightgray} 0.930  & 0.009  & 0.002 & 0.86 & 663.50 & 666.50 \\
            \hline
        \end{tabular}
    }
    \subtable[Voxelized method with texture size $227^3$.]{
        \begin{tabular}{| r | r | r | r | r |}
            \hline
            max(\# S) & max($\Delta E$) & mean($\Delta E$) & var($\Delta E$) & ray [ms]\\
            \hline
            70  & 79.283  & 6.609  & 157.129& 1.00 \\
            139  & 79.302  & 6.111  & 137.939& 1.58 \\
            277  & 78.049  & 5.616  & 118.935& 2.60 \\
            553  & 77.923  & 5.278  & 106.801& 4.57 \\
            1106  & 77.809  & 5.075  & 99.912 & 8.33 \\
            2211  & 77.726  & 4.962  & 96.138  & 15.38 \\
            4422  & 77.681  & 4.901  & 94.120  & 29.41 \\
            8844  & 75.615  & 4.865  & 92.794  & 55.56 \\
            \hline
        \end{tabular}
    }
    \subtable[ODE based method RK2 (second order).]{
        \begin{tabular}{| r | r | r | r | r | r | r | r |}
            \hline
            max(\# S) & max($\Delta P$) & max($\Delta E$) & mean($\Delta E$) & var($\Delta E$) & surf [ms] & ray [ms] & tot [ms] \\
            \hline
            11  &  20.4  & 54.582  & 3.066  & 22.910 & 0.86 & 5.87 & 6.99 \\
            23  &  4.8  & 35.191  & 1.433  & 7.277   & 0.86 & 9.43 & 10.57 \\
            47  &  1.1  & 17.771  & 0.531  & 1.075  & 0.86 & 16.15 & 17.30 \\
            95  & \cellcolor{lightgray} 0.6  &\cellcolor{lightgray} 4.745  & 0.186  & 0.108    & 0.86 & 29.43 & 30.67 \\
            190  & \cellcolor{lightgray} 0.6  &\cellcolor{lightgray} 2.110  & 0.076  & 0.022   & 0.86 & 55.33 & 56.78 \\
            381  & \cellcolor{lightgray} 0.6  &\cellcolor{lightgray} 1.167  & 0.038  & 0.012  & 0.86 & 105.80 & 107.40 \\
            762  & \cellcolor{lightgray} 0.6  &\cellcolor{lightgray} 1.091  & 0.020  & 0.006  & 0.86 & 203.10 & 205.70 \\
            1525  & \cellcolor{lightgray} 0.6  &\cellcolor{lightgray} 0.929  & 0.009  & 0.002 & 0.86 & 408.50 & 411.10 \\
            \hline
        \end{tabular}
    }
    \subtable[Voxelized method with texture size $341^3$.]{
        \begin{tabular}{| r | r | r | r | r |}
            \hline
            max(\# S) & max($\Delta E$) & mean($\Delta E$) & var($\Delta E$) & ray [ms] \\
            \hline
            70  & 79.213  & 5.760  & 126.757& 1.11 \\
            139  & 78.936  & 5.046  & 99.949 & 1.78 \\
            277  & 77.219  & 4.449  & 80.760 & 2.92 \\
            553  & 74.998  & 4.033  & 67.321 & 5.08 \\
            1106  & 74.790  & 3.777  & 59.694 & 9.26 \\
            2211  & 74.753  & 3.638  & 55.761  & 17.24 \\
            4422  & 74.800  & 3.563  & 53.635  & 33.33 \\
            8844  & 74.612  & 3.523  & 52.441  & 62.50 \\
            \hline
        \end{tabular}
    }
    \subtable[Root finding based (RF).]{
        \begin{tabular}{| r | r | r | r | r | r | r | r |}
            \hline
            max(\# S) & max($\Delta P$) & max($\Delta E$) & mean($\Delta E$) & var($\Delta E$) & surf [ms] & ray [ms] & tot [ms] \\
            \hline
            12  & 1.4  & 54.053  & 3.404  & 24.938  & 0.86 & 4.95 & 6.06  \\
            24  & \cellcolor{lightgray} 0.8  & 38.387  & 1.667  & 8.329  & 0.86 & 8.86 & 10.00 \\
            48  & \cellcolor{lightgray} 0.6  & 16.046  & 0.638  & 1.080  & 0.86 & 16.37 & 17.52 \\
            96  & \cellcolor{lightgray} 0.6  & 7.263  & 0.246  & 0.123   & 0.86 & 31.77 & 33.01 \\
            192  & \cellcolor{lightgray} 0.6  &\cellcolor{lightgray} 3.394  & 0.119  & 0.030   & 0.86 & 61.80 & 62.15 \\
            384  & \cellcolor{lightgray} 0.6  &\cellcolor{lightgray} 2.138  & 0.074  & 0.018  & 0.86 & 121.40 & 123.00 \\
            768  & \cellcolor{lightgray} 0.6  &\cellcolor{lightgray} 1.105  & 0.055  & 0.013  & 0.86 & 242.60 & 244.70 \\
            1536  & \cellcolor{lightgray} 0.6  &\cellcolor{lightgray} 0.962  & 0.035  & 0.006 & 0.86 & 484.50 & 487.10 \\
            \hline
        \end{tabular}
    }
    \subtable[Voxelized method with texture size $512^3$.]{
        \begin{tabular}{| r | r | r | r | r |}
            \hline
            max(\# S) & max($\Delta E$) & mean($\Delta E$) & var($\Delta E$) & ray[ms] \\
            \hline
            70  & 75.592  & 5.140  & 105.246& 1.54 \\
            139  & 74.997  & 4.245  & 73.714 & 2.65 \\
            277  & 74.956  & 3.581  & 55.287 & 4.03 \\
            553  & 74.624  & 3.133  & 43.332 & 6.41 \\
            1106  & 74.642  & 2.860  & 36.356 & 11.49 \\
            2211  & 72.952  & 2.701  & 32.473  & 21.28 \\
            4422  & 72.790  & 2.624  & 30.659  & 41.67 \\
            8844  & 72.414  & 2.585  & 29.749  & 83.33 \\
            \hline
        \end{tabular}
    }
    \caption{
    Statistics for the example shown left in \reffig{teaser} (parametrization quality for the twisted bar).
    max(\# S): maximum number of sample points;
    max($\Delta P$): largest error of pixel-accuracy among all pixels of the object;
    max/mean/var($\Delta E$): largest/mean of/variance of the color difference of a reference image;
    surf [ms]: time for creating ray-surface intersections;
    ray [ms]: time for blockwise depth-sorting and volume rendering for the proposed methods, time for standard volume rendering for "voxelized" method (texture is precomputed);
    tot [ms]: total render time.
    }
    \protect \label{table:accuracy}
\end{table*}

\begin{figure}
\centering
    \centering
    \subfigure[ODE RKF4 3/8 max 95 sample points]{\includegraphics[width=0.22\textwidth]{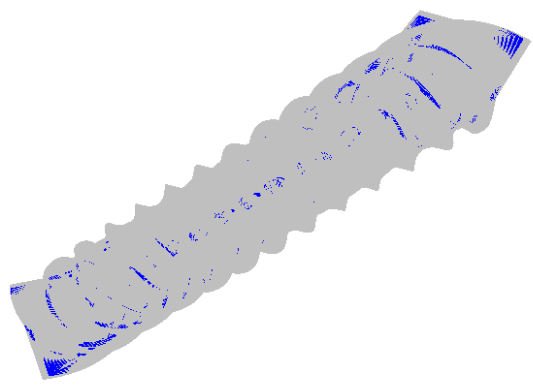}}
    \subfigure[ODE RKF4 3/8 max 95 sample points]{\includegraphics[width=0.22\textwidth]{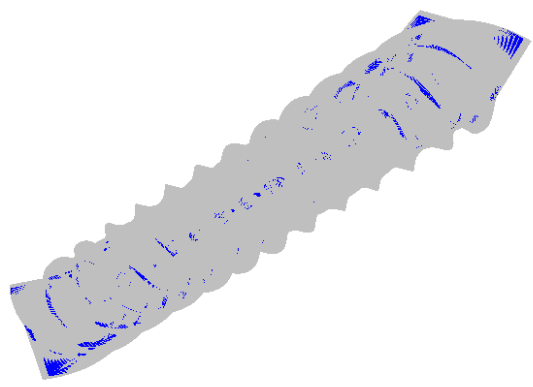}}
    \subfigure[RF max 96 sample points]{\includegraphics[width=0.22\textwidth]{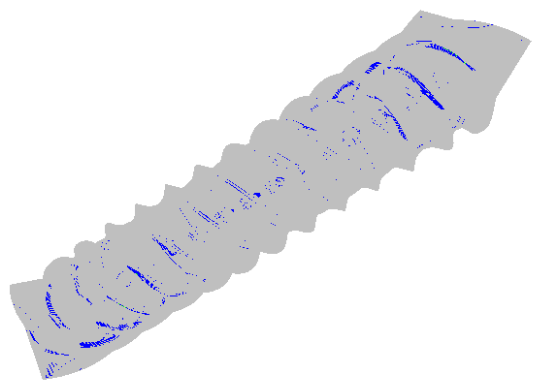}}
    \subfigure[Voxelized $512^3$, 1106 sample points]{\includegraphics[width=0.22\textwidth]{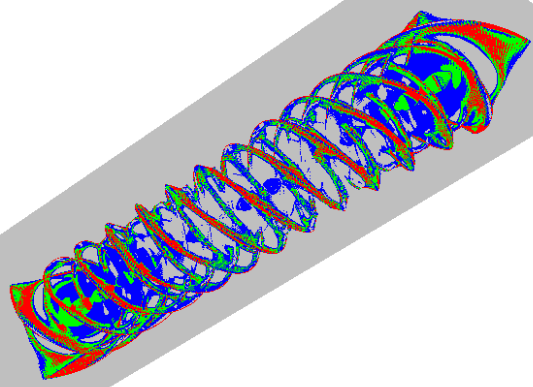}}
    \caption{
        Visualization of color difference to a reference solution.
        Each pixel is colored [grey, blue, green] if $\Delta E$ is at most [1,5,10].
        If it is red, the color difference is greater than 10.
        $\Delta E$ is the CIEDE2000 definition of color difference.
        Standard volume rendering of a voxelized version shows large errors even though 10 times as many samples are used.
    }
    \protect \label{fig:deltaEscaled}
\end{figure}

\subsection{Parametrization Quality of Geometry}

Before an analysis of a model can be carried out, the geometric shape has to be designed.
Since the parametrization of the geometry $\phi$ is not unique, one wants to optimize the quality of the parametrization.
A measure of the quality of the parametrization is given by
\begin{equation} \label{eq:paramQuality}
    \rho = \frac{\mbox{det}(\Jac{\phi})}{\|\Jac{\phi}\|_F},
\end{equation}
where low values indicate that the geometry is (close to) degenerate.
With this scalar field our method can be used as an inspection tool in the design phase, isolating potentially problematic areas.
The parametrization quality \eqref{eq:paramQuality} is calculated on the fly for each sample.

As an example we present a twisted bar, see left in \reffig{teaser}.
This model consists of only one volume block and the geometry is given by a quadratic B-spline with 425 control points.
We compare standard ray-casting of the pre-computed voxelized method with the proposed methods on a screen resolution of $640\times480$ pixels.
The first difference is that the proposed method uses 63 MB, and the voxelized method uses [67,178,528] MB for a texture size of [$227^3$,$341^3$,$512^3$].
For the proposed methods we can see in \reftable{accuracy} that the parametric accuracy $\Delta P$ (as defined in \refequ{deltaP}), decreases to 0.6 as the number of samples increases.
It does not decrease further because the $\Delta P = 0.6$ is already reached for the view-ray intersections on the surfaces.
We can also see that the color difference $\Delta E$ decreases with the number of sample points.

For the standard volume rendering of the precomputed ("voxelized") version of the model, the notion of pixel-accuracy is not meaningful.
Ultimately, one is interested in the color difference $\Delta E$.
\reftable{accuracy} shows that $\Delta E$ decreases with the number of sample points (for the volume render integral).
The color difference also decreases with increased texture size.
Naturally, for the same number of sample points along the view-rays, the volume rendering of the voxelized model is much faster than the proposed methods.
However, for a texture size of $512^3$ with almost 9000 sample points, the maximum color difference is still around 72, and the mean is larger than 2.
We can see in \reffig{deltaEscaled} that the color difference is highest at the boundary of the model, but for the voxelized model even the interior points show values between 5 and 10, meaning a noticeable difference to the reference image.

Overall, \reftable{performance} shows that the second order ODE based method performs best on this model.
Higher order ODE based methods typically allow larger sample distances while still being pixel-accurate.
In this example all ODE based methods, except the first order methods, have the same sample distance which is dictated by the volume render integral.
Therefore the second order midpoint method is fastest.

\subsection{Stress Analysis in Linear Elasticity}\label{sec:terrific}

\setlength{\tabcolsep}{.12em}
\begin{table*}
    \centering
    \subtable[Parametrization quality for the twisted bar, see left in \reffig{teaser}.]{
        \begin{tabular}{| l | r |}
            \hline
            method &  tot [ms] \\
            \hline
            RK 1 (c=1) & 332 \\
            IRK 1 (c=100) & 204 \\
            RK 2 (c=1) & 31 \\
            RK 3 (c=1) & 37 \\
            RK 4 (c=1) & 46 \\
            RK 4 3/8 (c=1) & 48 \\
            RKF (c=1) & 68 \\
            RF & 47 \\
            \hline
        \end{tabular}
    }
    \subtable[Von Mises stress for TERRIFIC model, see middle in \reffig{teaser}.]{
        \begin{tabular}{| l | r |}
            \hline
            method &  tot [ms] \\
            \hline
            RK 1 (c=1) & 354 \\
            IRK 1 (c=100) & 321 \\
            RK 2 (c=1) & 78 \\
            RK 3 (c=1) & 93 \\
            RK 4 (c=1) & 108 \\
            RK 4 3/8 (c=1) & 109 \\
            RKF (c=1) & 137 \\
            RF & 73 \\
            \hline
        \end{tabular}
    }
    \subtable[Backstep Flow from RANS simulation, see right in \reffig{teaser}]{
        \begin{tabular}{| l | r |}
            \hline
            method &  tot [ms] \\
            \hline
            RK 1 (c=1) & 41 \\
            IRK 1 (c=100) & 76 \\
            RK 2 (c=1) & 52 \\
            RK 3 (c=1) & 67 \\
            RK 4 (c=1) & 83 \\
            RK 4 3/8 (c=1) & 82 \\
            RKF (c=1) & 111 \\
            RF & 62 \\
            \hline
        \end{tabular}
    }
    \caption{
        Comparison of the performance of the proposed methods for visualization of different models on an NVIDIA Titan GPU and a screen resolution of $640\times 480$.
        All methods use the largest (uniform) sample distance of the volume render integral, but are at the same time pixel accurate, i.e., $\Delta P\leq 1$ (as defined in \refequ{deltaP}) and the color difference to a reference image is $\Delta E \leq 5$.
    }
    \protect \label{table:performance}
\end{table*}
Structural analysis is an important application area for isogeometric analysis, where external forces lead to a deformation of the object given in the form of a so-called displacement field
$u: P \rightarrow \R^3$.
The stress due to deformation is then calculated by
\begin{align}\label{eq:vonMises}
  \begin{split}
 \rho = \Big(&    (\sigma_{11}-\sigma_{22}) + (\sigma_{22}-\sigma_{33}) + (\sigma_{33}-\sigma_{11})\\
             & + 6(\sigma_{12}^2 + \sigma_{23}^2 + \sigma_{31}^2) \Big)/2.
  \end{split}
\end{align}
The so-called strain tensor is
$ \sigma = \frac{1}{2}( \Jac{u\circ\phi^{-1}} + \Jac{u\circ\phi^{-1}}^T ) $,
where $\Jac{u\circ\phi^{-1}} (g)$ is the solution of
$$ \left[ \Jac{\phi}(p) \right]^T \Jac{u\circ\phi^{-1}}(g) = \left[ \Jac{u}(p) \right]^T. $$
As in the previous example, all those expressions are calculated on the fly for each sample of the scalar field.

In the middle of \reffig{teaser} we present the results for the linear elasticity simulation from the TERRIFIC project.
Both geometry $\phi$ and deformation $u$ are cubic B-splines and the model consists of 15 volume blocks with 2484 control points.
All the proposed methods work well also in this case where there are some degeneracies along the boundaries, see \reffig{DemonstratorQuality}.
As \reftable{performance}~(b) indicates, the root finding based method and the second order midpoint method have the fastest render times.

\subsection{Computational Fluid Dynamics}

Another important application of isogeometric analysis is computational fluid dynamics (CFD).
On the right of \reffig{teaser} we present a visualization of an approximation of the solution of the Reynolds-averaged Navier-Stokes (RANS) equations for a backstep flow.
The scalar field $\rho$ represents turbulent viscosity and comes directly from the simulation.
The model uses quadratic B-splines to represent both the geometry and the scalar field and consists of 140 blocks with 888642 control points.
As can be seen from \reftable{performance} (c) the explicit Euler method has the fastest render time, followed by the second order midpoint method and the root finding based method.
Since the geometry and the scalar field is close to linear, the first order ODE solver is the most efficient compared with higher order ODE solvers.

\section{Conclusion}\label{sec:conclusion}

The presented approach allows interactive inspection of volumetric models used in isogeometric analysis.
In the spirit of isogeometry, the algorithms operate directly on the spline models and therefore demand very little GPU memory.
The proposed algorithms enable pixel-accurate geometry of both surfaces and volume irrespective of the zoom level, making it an asset during the design, analysis and marketing phase.
We applied our approach in three use cases relevant to industry showing good performance at interactive frame rates.

In the future we plan to develop algorithms that automatically choose a sample distance that ensures pixel-accuracy.
In addition, we seek to increase the efficiency of the presented methods by developing algorithms for adaptive sampling, as well as exploring methods for automatically choosing the order of the ODE based methods.

\begin{figure}
    \centering
    \includegraphics[width=0.99\linewidth]{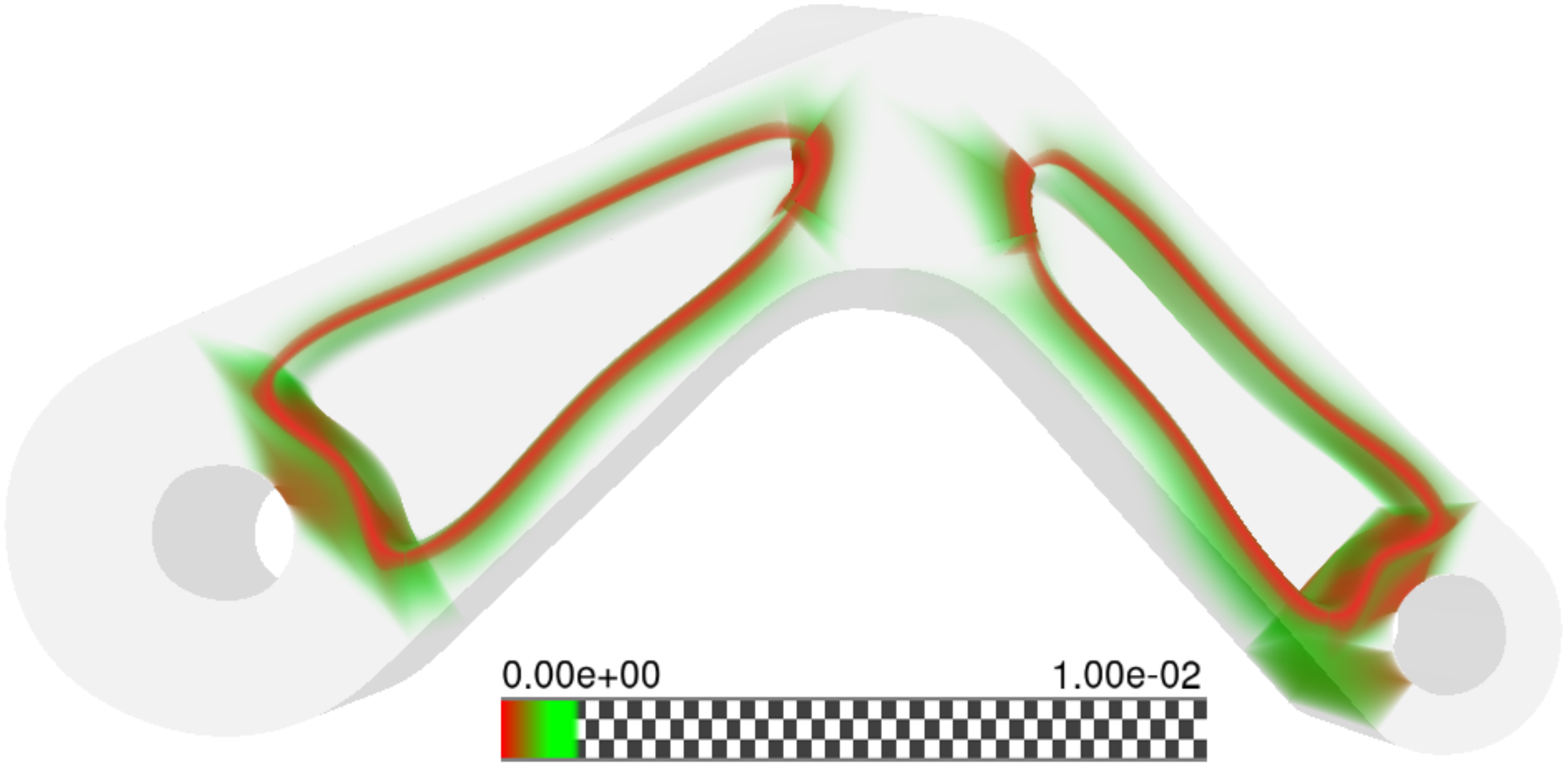}
    \caption{
        Quality of parametrization (see \refequ{paramQuality}) for the model from the TERRIFIC project see middle of \reffig{teaser}.
        Low values shown in red indicate areas with a problematic parametrization.
        The checkerboard pattern indicates alpha values less than 1.
    }
    \protect \label{fig:DemonstratorQuality}
\end{figure}

\newpage

\bibliography{isogeometric_volume_visualization}

\begin{biography}[{\includegraphics[width=1in,
    height=1.25in,clip,keepaspectratio]{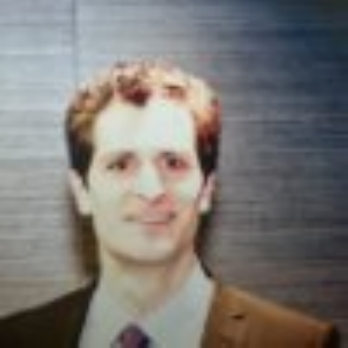}}]{Franz G. Fuchs}
    received his master's degree (Diplom) in mathematics from the Technical University of Munich (TUM) in 2006 with a thesis on image processing.
    In 2009 he received his PhD in applied mathematics from the University of Oslo (CMA), working on mathematical theory and numerical methods for hyperbolic conservation laws.
    His additional research interests include efficient numerical algorithms on parallel architectures for visualization and computation.
\end{biography}

\begin{biography}[{\includegraphics[width=1in,
    height=1.25in,clip,keepaspectratio]{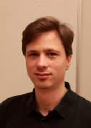}}]{Jon M. Hjelmervik}
    received his PhD in cotutelle between the University of Oslo and Grenoble INP in 2009.
    He is a research manager associated to SINTEF ICT Applied Mathematics since 1998.
    Until spring 2010 he held a 20\% position as associate professor at Norwegian School of Information Technology (NITH).
     His research interests include visualization of isogeometric representations and heterogeneous computing in cloud-based frameworks.  
\end{biography}
\vfill

\end{document}

%% file: tikzdefines.tex
\usepackage{tikz}
\usetikzlibrary{calc, matrix,decorations.text, decorations.pathmorphing}
\usepackage{tikz-3dplot}

\global\def\amp{0.75}
\global\def\shiftinvgamma{.2}
\global\def\freq{57./3.*2.*3.14159265359}
\def\phiOne(#1,#2){{#2+(\amp*sin(\freq*#1))}}
\def\phiTwo(#1,#2){{1.125+#2-(#1-1.5)*(#1-1.5)/2.}}
\def\phiInv(#1){{-.7+1.4*\amp*sin(\freq*#1)}}
\def\phiInvApprox(#1,#2,#3){{1.1-.7*((#1-#2+.01)/#3)^.2}}
\global\def\radius{0.1} 

\global\def\translatePAR{6}
\global\def\translateGEO{0}
\global\def\translateRHO{1.5}

\newcommand{\drawRayCastingISOsetup}[0]{
    \def\ang{335}
    \pgfmathsetmacro{\TAx}{{0}}
    \pgfmathsetmacro{\TAy}{{0-1.5}}
    \pgfmathsetmacro{\TAz}{{3}}
    \pgfmathsetmacro{\TBx}{{0}}
    \pgfmathsetmacro{\TBy}{{2-1.5}}
    \pgfmathsetmacro{\TBz}{{3}}

    \pgfmathsetmacro{\TCx}{{4}}
    \pgfmathsetmacro{\TCy}{{2-1.5}}
    \pgfmathsetmacro{\TCz}{{3}}
    \pgfmathsetmacro{\TDx}{{4}}
    \pgfmathsetmacro{\TDy}{{0-1.5}}
    \pgfmathsetmacro{\TDz}{{3}}

    \pgfmathsetmacro{\Ax}{{cos(\ang)*\TAx+sin(\ang)*\TAz}}
    \pgfmathsetmacro{\Ay}{\TAy}
    \pgfmathsetmacro{\Az}{{-sin(\ang)*\TAx+cos(\ang)*\TAz}}
    \pgfmathsetmacro{\Bx}{{cos(\ang)*\TBx+sin(\ang)*\TBz}}
    \pgfmathsetmacro{\By}{\TBy}
    \pgfmathsetmacro{\Bz}{{-sin(\ang)*\TBx+cos(\ang)*\TBz}}

    \pgfmathsetmacro{\Cx}{{cos(\ang)*\TCx+sin(\ang)*\TCz}}
    \pgfmathsetmacro{\Cy}{\TCy}
    \pgfmathsetmacro{\Cz}{{-sin(\ang)*\TCx+cos(\ang)*\TCz}}
    \pgfmathsetmacro{\Dx}{{cos(\ang)*\TDx+sin(\ang)*\TDz}}
    \pgfmathsetmacro{\Dy}{\TDy}
    \pgfmathsetmacro{\Dz}{{-sin(\ang)*\TDx+cos(\ang)*\TDz}}
    \pgfmathsetmacro{\halfx}{{\Ax+0.5*(\Cx-\Ax)}}
    \pgfmathsetmacro{\halfy}{{\Ay+0.5*(\Cy-\Ay)}}
    \pgfmathsetmacro{\halfz}{{\Az+0.5*(\Cz-\Az)}}

    \pgfmathsetmacro{\backx}{{.6}}
    \pgfmathsetmacro{\backy}{{1}}
    \pgfmathsetmacro{\backz}{{-2.4}}

}

\newcommand{\drawBadApprox}[0]{

    \def\pa{{\translatePAR+\Pt +\shiftinvgamma + .64},{ .4}}
    \def\pb{{\translatePAR+\Pt +\shiftinvgamma + .04},{+.09}}
    \def\pc{{\translatePAR+\Pt +\shiftinvgamma + .00},{-.5}}

    \filldraw[ball color=red] ({\pa}) circle (\radius);
    \filldraw[ball color=red] ({\pb}) node[right=8pt, below=1pt] {$\pmb{\widetilde p_i}$} circle (\radius);
    \filldraw[ball color=red] ({\pc}) circle (\radius);
    \draw[very thick,red,densely dashed] ({\pa}) -- ({\pb}) -- ({\pc});

    \def\tlin{0.75}
    \def\ga{{\ox+\tlin*(\lx-\ox)},{\oy+\tlin*(\ly-\oy)}}

    \def\tlout{0.9}
    \def\gc{{\ox+\tlout*(\lx-\ox)},{\oy+\tlout*(\ly-\oy)}}

    \def\tmpb{\oy+\tlout*(\ly-\oy)-.6}
    \def\gb{{\ox+\tlout*(\lx-\ox)+.0},{\tmpb}}

	\filldraw[ball color=red] ({\ga}) circle (\radius);
	\filldraw[ball color=red] ({\gb}) node[right=5pt,below=5pt] {$\pmb{\widetilde g_i}$} circle (\radius) ;
	\filldraw[ball color=red] ({\gc}) circle (\radius);
    \draw[very thick, shorten >=1pt, shorten <=1pt,-triangle 45] ({\pa}) to node[above] {$ \quad \quad \pmb{\phi}$} ({\gc});
    \draw[very thick, shorten >=1pt, shorten <=1pt,-triangle 45] ({\pb}) -- ({\gb});
    \draw[very thick, shorten >=1pt, shorten <=1pt,-triangle 45] ({\pc}) -- ({\ga});
}

\newcommand{\drawFrustum}[0]{
    \def\ax{1.45};\def\ay{1.82};
    \def\bx{0.7};\def\by{2.0};
    \def\cx{1.45};\def\cy{0.95};
    \def\dx{0.7};\def\dy{1.15};

    \def\ta{0.38};
    \def\tb{0.38};
    \def\tc{0.37};
    \def\td{0.43};
    \draw[thin,black,-] ({\ox+\ta*(\ax-\ox)},{\oy+\ta*(\ay-\oy)}) -- ({\ax},{\ay});
    \draw[thin,black,-] ({\ox+\tb*(\bx-\ox)},{\oy+\tb*(\by-\oy)}) -- ({\bx},{\by});
    \draw[thin,black,-] ({\ox+\tc*(\cx-\ox)},{\oy+\tc*(\cy-\oy)}) -- ({\cx},{\cy});
    \draw[thin,black,-] ({\ox+\td*(\dx-\ox)},{\oy+\td*(\dy-\oy)}) -- ({\dx},{\dy});

    \def\tdo{0.361};
    \draw[thin,gray,-] ({\ox+\tdo*(\dx-\ox)},{\oy+\tdo*(\dy-\oy)})-- ({\ox+\td*(\dx-\ox)},{\oy+\td*(\dy-\oy)});
    \draw[cube,black,thin,fill=blue,opacity=.2] ({\ax},{\ay}) -- ({\bx},{\by}) -- ({\dx},{\dy}) -- ({\cx},{\cy}) -- cycle;

    \def\tab{0.78};
    \def\tbb{0.78};
    \def\tcb{0.78};
    \def\tdb{0.78};
    \draw[cube,black,thin,fill=blue,opacity=.2] ({\ox+\tab*(\ax-\ox)},{\oy+\tab*(\ay-\oy)}) --
    ({\ox+\tbb*(\bx-\ox)},{\oy+\tbb*(\by-\oy)}) --
    ({\ox+\tdb*(\dx-\ox)},{\oy+\tdb*(\dy-\oy)}) --
    ({\ox+\tcb*(\cx-\ox)},{\oy+\tcb*(\cy-\oy)}) --
      cycle;

    \def\tab{\ta};
    \def\tbb{\tb};
    \def\tcb{\tc};
    \def\tdb{\td};
    \draw[cube,black,thin,fill=blue,opacity=.2] ({\ox+\tab*(\ax-\ox)},{\oy+\tab*(\ay-\oy)}) --
    ({\ox+\tbb*(\bx-\ox)},{\oy+\tbb*(\by-\oy)}) --
    ({\bx},{\by}) --
    ({\ax},{\ay}) --
      cycle;

    \draw[cube,black,thin,fill=blue,opacity=.2] ({\ox+\tab*(\ax-\ox)},{\oy+\tab*(\ay-\oy)}) --
    ({\ox+\tcb*(\cx-\ox)},{\oy+\tcb*(\cy-\oy)}) --
    ({\cx},{\cy}) --
    ({\ax},{\ay}) --
      cycle;

}

\newcommand{\drawRayCastingISOrayingeoSurface}[0]{
    \def\tlin{0.75}
	\filldraw[ball color=blue] ({\ox+\tlin*(\lx-\ox)},{\oy+\tlin*(\ly-\oy)}) circle (\radius);

    \def\tlout{0.9}
	\filldraw[ball color=blue] ({\ox+\tlout*(\lx-\ox)},{\oy+\tlout*(\ly-\oy)}) circle (\radius);
}

\newcommand{\drawEye}[3]{
    \pgfmathsetmacro{\ex}{#1}
    \pgfmathsetmacro{\ey}{#2}
    \pgfmathsetmacro{\eyeSize}{.9}
    \pgfmathsetmacro{\eRot}{#3}
    \pgfmathsetmacro{\eAp}{-55}
    \draw[rotate around={\eRot:(\ex,\ey)}] (\ex,\ey) -- ++(-.5*\eAp:\eyeSize)
                                           (\ex,\ey) -- ++(.5*\eAp:\eyeSize);
    \draw (\ex,\ey) ++(\eRot+\eAp:.75*\eyeSize) arc (\eRot+\eAp:\eRot-\eAp:.75*\eyeSize);

    \draw[fill=gray] (\ex,\ey) ++(\eRot+\eAp/3:.75*\eyeSize) 
                 arc (\eRot+180-\eAp:\eRot+180+\eAp:.28*\eyeSize);

    \draw[fill=black] (\ex,\ey) ++(\eRot+\eAp/3:.75*\eyeSize) 
                  arc (\eRot+\eAp/3:\eRot-\eAp/3:.75*\eyeSize);
}

\newcommand{\drawRayCastingISOscreen}[0]{

    \def\ox{-1.75};\def\oy{-3.9};

    \def\tl{0.1}
    \def\lx{1.35};\def\ly{1.8};
    \draw[thick,blue,densely dashed,-triangle 45] ({\ox+\tl*(\lx-\ox)},{\oy+\tl*(\ly-\oy)}) -- ({\lx},{\ly});
    \def\tlin{0.75}
    \def\tlout{0.9}
    \draw[thick,blue,-] ({\ox+\tlin*(\lx-\ox)},{\oy+\tlin*(\ly-\oy)}) -- ({\ox+\tlout*(\lx-\ox)},{\oy+\tlout*(\ly-\oy)});
    \drawEye{\translateGEO-1.75}{-3.9}{60}

    \pgfmathsetmacro{\xy}{{5}}
    \pgfmathsetmacro{\xx}{{2}}
    \pgfmathsetmacro{\Ahalf}{{\Ay+\xx/5*(\By-\Ay)}}
    \pgfmathsetmacro{\Ahalfp}{{\Ay+(\xx+1)/5*(\By-\Ay)}}
    \pgfmathsetmacro{\Dhalf}{{\Dy+\xx/5*(\Cy-\Dy)}}
    \pgfmathsetmacro{\Dhalfp}{{\Ay+(\xx+1)/5*(\By-\Ay)}}

    \pgfmathsetmacro{\Axhalf}{{\Ax+\xy/10*(\Dx-\Ax)}}
    \pgfmathsetmacro{\Ayhalf}{{\Ahalf+\xy/10*(\Dhalf-\Ahalf)}}
    \pgfmathsetmacro{\Azhalf}{{\Az+\xy/10*(\Dz-\Az)}}

    \pgfmathsetmacro{\Axhalfp}{{\Ax+\xy/10*(\Dx-\Ax)}}
    \pgfmathsetmacro{\Ayhalfp}{{\Ahalfp+\xy/10*(\Dhalfp-\Ahalfp)}}
    \pgfmathsetmacro{\Azhalfp}{{\Az+\xy/10*(\Dz-\Az)}}

    \pgfmathsetmacro{\Dxhalf}{{\Ax+(\xy+1)/10*(\Dx-\Ax)}}
    \pgfmathsetmacro{\Dyhalf}{{\Ahalf+(\xy+1)/10*(\Dhalf-\Ahalf)}}
    \pgfmathsetmacro{\Dzhalf}{{\Az+(\xy+1)/10*(\Dz-\Az)}}

    \pgfmathsetmacro{\Dxhalfp}{{\Ax+(\xy+1)/10*(\Dx-\Ax)}}
    \pgfmathsetmacro{\Dyhalfp}{{\Ahalfp+(\xy+1)/10*(\Dhalfp-\Ahalfp)}}
    \pgfmathsetmacro{\Dzhalfp}{{\Az+(\xy+1)/10*(\Dz-\Az)}}
    \draw[cube,black,thin,fill=blue,opacity=.4] ({\translateGEO+\Axhalf},{\Ayhalf},{\Azhalf}) -- ({\translateGEO+\Axhalfp},{\Ayhalfp},{\Azhalfp}) -- ({\translateGEO+\Dxhalfp},{\Dyhalfp},{\Dzhalfp}) -- ({\translateGEO+\Dxhalf},{\Dyhalf},{\Dzhalf}) -- cycle;

    \def\Lb{1}
    \foreach \x in {1,2,3,4} {
        \pgfmathsetmacro{\Ahalf}{{\Ay+\x/5*(\By-\Ay)}}
        \pgfmathsetmacro{\Dhalf}{{\Dy+\x/5*(\Cy-\Dy)}}
        \draw[cube,black,thin,opacity=.5] ({\translateGEO+\Ax+\Lb/10*(\Dx-\Ax)},{\Ahalf+\Lb/10*(\Dy-\Ay)},{\Az+\Lb/10*(\Dz-\Az)}) -- ({\translateGEO+\Dx},{\Dhalf},{\Dz});
    }
    \foreach \x in {2,3,4,5,6,7,8,9} {
        \pgfmathsetmacro{\Axhalf}{{\Ax+\x/10*(\Dx-\Ax)}}
        \pgfmathsetmacro{\Ayhalf}{{\Ay+\x/10*(\Dy-\Ay)}}
        \pgfmathsetmacro{\Azhalf}{{\Az+\x/10*(\Dz-\Az)}}
        \pgfmathsetmacro{\Bxhalf}{{\Bx+\x/10*(\Cx-\Bx)}}
        \pgfmathsetmacro{\Byhalf}{{\By+\x/10*(\Cy-\By)}}
        \pgfmathsetmacro{\Bzhalf}{{\Bz+\x/10*(\Cz-\Bz)}}
        \draw[cube,black,thin,opacity=.5] ({\translateGEO+\Axhalf},{\Ayhalf},{\Azhalf}) -- ({\translateGEO+\Bxhalf},{\Byhalf},{\Bzhalf});
    }
    \draw[cube,opacity=.5] ({\translateGEO+\Ax+\Lb/10*(\Dx-\Ax)},{\Ay+\Lb/10*(\Dy-\Ay)},{\Az+\Lb/10*(\Dz-\Az)}) --
                           ({\translateGEO+\Bx+\Lb/10*(\Cx-\Bx)},{\By+\Lb/10*(\Cy-\By)},{\Bz+\Lb/10*(\Cz-\Bz)}) --
                           ({\translateGEO+\Cx},{\Cy},{\Cz}) --
                           ({\translateGEO+\Dx},{\Dy},{\Dz}) -- cycle;

}

\newcommand{\annotateInOut}[0]{
    \def\ox{-1.75};\def\oy{-3.9};
    \def\tlin{0.75}
    \node[right=10pt, below=0.1pt] (RayIn) at ({\ox+\tlin*(\lx-\ox)},{\oy+\tlin*(\ly-\oy)})  {$\gfront$};
    \def\tlout{0.9}
    \node[right=12pt, below=-4pt] (RayOut) at ({\ox+\tlout*(\lx-\ox)},{\oy+\tlout*(\ly-\oy)})  {$\gback$};

}

\newcommand{\drawRayCastingISOrayinparamSurface}[0]{

    \def\geoXIn{0.75}
    \def\geoYIn{.25}
    \def\geoZIn{2}

    \def\geoXOut{.75}
    \def\geoYOut{0.5}
    \def\geoZOut{0}

}

\newcommand{\drawRayCastingISOrayingeoVolumeDepth}[0]{
    \def\tl{0.1}
    \def\ox{-1.75};\def\oy{1.9};
    \def\lx{3.5};\def\ly{.9};
    \def\tr{1.6}
    \draw[thick,blue,->] ({\ox+\tl*(\lx-\ox)},{\oy+\tl*(\ly-\oy)}) -- ({\ox+\tr*(\lx-\ox)},{\oy+\tr*(\ly-\oy)});

    \def\tlA{0.25}
    \def\tlB{0.5}
    \draw[very thick,blue,-] ({\ox+\tlA*(\lx-\ox)},{\oy+\tlA*(\ly-\oy)}) -- ({\ox+\tlB*(\lx-\ox)},{\oy+\tlB*(\ly-\oy)});
    \def\tlA{0.75}
    \def\tlB{0.8}
    \draw[very thick,blue,-] ({\ox+\tlA*(\lx-\ox)},{\oy+\tlA*(\ly-\oy)}) -- ({\ox+\tlB*(\lx-\ox)},{\oy+\tlB*(\ly-\oy)});
    \foreach \tlin in {0.25,0.5,0.75,0.8,1.35} {
        \node[right] (RayIn) at ({\ox+\tlin*(\lx-\ox)},{\oy+\tlin*(\ly-\oy)})  {};
        \filldraw[ball color=blue] ({\ox+\tlin*(\lx-\ox)},{\oy+\tlin*(\ly-\oy)}) circle (\radius);
    }

    \drawEye{\ox}{\oy}{-10}

}

\usetikzlibrary{decorations.markings}

\tikzset{
    shaded/.style={
        postaction={
            decorate,
            decoration={
                markings,
                mark=between positions 0 and \pgfdecoratedpathlength step 0.5pt with {
                    \pgfmathsetmacro\myval{multiply(divide(
                    \pgfkeysvalueof{/pgf/decoration/mark info/distance from start}, \pgfdecoratedpathlength),120)};
                    \pgfsetfillcolor{green!\myval!red};
                    \pgfpathcircle{\pgfpointorigin}{#1};
                    \pgfusepath{fill};
                }
            }
        }
    }
}

\newcommand{\drawRayCastingISOrayinparamVolume}[0]{

    \def\a{0.1}
    \def\b{0.77}
    \def\sfs{-.1}
    \def\y{+.1}
    \draw [thick,blue] ({\translatePAR+\Pt+\shiftinvgamma+ \a-.13+\sfs},-0.5+\y) .. controls (\translatePAR+\Pt+\shiftinvgamma+\sfs-.3,0.4+\y) .. ({\translatePAR+\Pt+\shiftinvgamma+ \b-.13+\sfs},.35+\y);

    \node[right=3pt] (RayIn) at ({\translatePAR+\Pt+\shiftinvgamma+ \a-.13+\sfs},-0.5+\y) {$\pfront$};
    \node[right=10pt, above=2pt] (RayOut) at ({\translatePAR+\Pt+\shiftinvgamma+ \b-.13+\sfs},.35+\y) {$\pback$};

    \filldraw[ball color=blue] ({\translatePAR+\Pt+\shiftinvgamma+ \a-.13+\sfs},-0.5+\y) circle (\radius);

    \filldraw[ball color=blue] ({\translatePAR+\Pt+\shiftinvgamma+ \b-.13+\sfs},.35+\y) circle (\radius);

}

\def\phiInvDepth(#1,#2,#3){{(#3+#2*#1*#1}}

\newcommand{\drawRayCastingISOrayinparamVolumeDepthOne}[1]{

    \draw[thick,blue,domain=-.39:.39] plot[smooth] ({\translatePAR+\Pt+ \x+.2+2},{\phiInvDepth(\x,4,1+#1)});
    \draw[thick,blue,domain=-.39:.0] plot[smooth] ({\translatePAR+\Pt+ \x+1.6+2},{\phiInvDepth(\x,2,1.3+#1)});
    \filldraw[ball color=blue] ({\translatePAR+\Pt -.39+.2+2},{\phiInvDepth(-.39,4,1+#1)}) circle (\radius);
    \filldraw[ball color=blue] ({\translatePAR+\Pt+ .39+.2+2},{\phiInvDepth(+.39,4,1+#1)}) circle (\radius);
    \filldraw[ball color=blue] ({\translatePAR+\Pt- .39+1.6+2},{\phiInvDepth(-.39,2,1.3+#1)}) circle (\radius);
    \filldraw[ball color=blue] ({\translatePAR+\Pt+ 0. +1.6+2},{\phiInvDepth(0,2,1.3+#1)}) circle (\radius);

}

\newcommand{\drawRayCastingISOrayinparamVolumeDepthTwo}[1]{

    \draw[thick,blue,domain=-.6:1.3] plot[smooth] ({\translatePAR+\Pt+ \x+.2+2},{(\x-.5)*(\x-.5)*.5+.5+#1)});
    \filldraw[ball color=blue] ({\translatePAR+\Pt -.39+2},{\phiInvDepth(-.39,4,1+#1-.5)}) circle (\radius);
    \filldraw[ball color=blue] ({\translatePAR+\Pt+ 0. +1.6+2},{\phiInvDepth(0,2,1.3+#1-.42)}) circle (\radius);

}

\newcommand{\drawRayCastingMultiblock}[2]{

    \draw[thick,dashed,black] (\translateGEO+0,0,0) -- (\translateGEO+0,2,0) -- (\translateGEO+0,2,2);
    \draw[thick,black] (\translateGEO+0,0,0) -- (\translateGEO+0,0,2);
    \draw[thick,black] (\translateGEO+0,0,2) -- (\translateGEO+0,2,2);
    \draw[thick,dashed,black] (\translateGEO+2+1,0,2) -- (\translateGEO+2+1,2,2) -- (\translateGEO+2+1,2,0)-- (\translateGEO+2+1,0,0) -- cycle;

    \draw[thick,dashed,black,domain=0.:2.3] plot[smooth] (\translateGEO+\x,{\phiOne(\x,0)},0);
    \draw[thick,black,domain=2.2:3] plot[smooth] (\translateGEO+\x,{\phiOne(\x,0)},0);

    \draw[thick,black,domain=0.:3.] plot[smooth] (\translateGEO+\x,{\phiOne(\x,0)},2);

    \draw[thick,dashed,black,domain=2.3:3.] plot[smooth] (\translateGEO+\x,{\phiOne(\x,2)},0);
    \draw[thick,black,domain=0:2.2] plot[smooth] (\translateGEO+\x,{\phiOne(\x,2)},0);

    \draw[thick,black,domain=0.:3.] plot[smooth] (\translateGEO+\x,{\phiOne(\x,2)},2);

    \node (P1) at (\translatePAR+1+\Pt,0.5+#1) {$P_1$};
    \node (G1) at (\translateGEO+1.5,0.5+#1) {$G_1$};
    \draw[thick,-triangle 45] (P1) to node[above] {$\phi_1$} (G1);

    \def\translateGEO{3}
    \draw[thick,black] (\translateGEO+0,0,2) -- (\translateGEO+0,2,2);
    \draw[thick,black] (\translateGEO+2+1,0,2) -- (\translateGEO+2+1,2,2) -- (\translateGEO+2+1,2,0)-- (\translateGEO+2+1,0,0) -- cycle;

    \draw[thick,dashed,black,domain=0.:2.3] plot[smooth] (\translateGEO+\x,{\phiTwo(\x,0)},0);
    \draw[thick,black,domain=2.2:3] plot[smooth] (\translateGEO+\x,{\phiTwo(\x,0)},0);

    \draw[thick,black,domain=0.:3.] plot[smooth] (\translateGEO+\x,{\phiTwo(\x,0)},2);

    \draw[thick,black,domain=0.1:3.] plot[smooth] (\translateGEO+\x,{\phiTwo(\x,2)},0);
    \draw[thick,dashed,black,domain=0:0.1] plot[smooth] (\translateGEO+\x,{\phiTwo(\x,2)},0);

    \draw[thick,black,domain=0.:3.] plot[smooth] (\translateGEO+\x,{\phiTwo(\x,2)},2);

    \foreach \mv in {#1, #2} {
        \draw[cube hidden,black] (\translatePAR+2+\Pt+2,0+\mv,0) -- (\translatePAR+0+\Pt+2,0+\mv,0) -- (\translatePAR+0+\Pt+2,2+\mv,0);
        \draw[cube hidden,black] (\translatePAR+0+\Pt+2,0+\mv,0) -- (\translatePAR+0+\Pt+2,0+\mv,2);
        \draw[cube,black] (\translatePAR+0+\Pt+2,0+\mv,2) -- (\translatePAR+0+\Pt+2,2+\mv,2) -- (\translatePAR+2+\Pt+2,2+\mv,2) -- (\translatePAR+2+\Pt+2,0+\mv,2) -- cycle;
        \draw[cube,black] (\translatePAR+2+\Pt+2,2+\mv,2) -- (\translatePAR+2+\Pt+2,0+\mv,2) -- (\translatePAR+2+\Pt+2,0+\mv,0) -- (\translatePAR+2+\Pt+2,2+\mv,0) -- cycle;
        \draw[cube,black] (\translatePAR+2+\Pt+2,2+\mv,2) -- (\translatePAR+0+\Pt+2,2+\mv,2) -- (\translatePAR+0+\Pt+2,2+\mv,0) -- (\translatePAR+2+\Pt+2,2+\mv,0) -- cycle;
    }

    \node (P2) at (\translatePAR+1+\Pt,0.5+#2-1.5) {$P_2$};
    \node (G2) at (\translateGEO+0.5,0.5+#2-1.5) {$G_2$};
    \draw[thick,-triangle 45] (P2) to node[above] {$\phi_2$} (G2);

}

\newcommand{\drawRayCastingISObase}[0]{

    \draw[thick,dashed,black] (\translateGEO+0,0,0) -- (\translateGEO+0,2,0) -- (\translateGEO+0,2,2);
    \draw[thick,black] (\translateGEO+0,0,0) -- (\translateGEO+0,0,2);
    \draw[thick,black] (\translateGEO+0,0,2) -- (\translateGEO+0,2,2);
    \draw[thick,black] (\translateGEO+2+1,0,2) -- (\translateGEO+2+1,2,2) -- (\translateGEO+2+1,2,0)-- (\translateGEO+2+1,0,0) -- cycle;

    \draw[thick,dashed,black,domain=0.:2.3] plot[smooth] (\translateGEO+\x,{\phiOne(\x,0)},0);
    \draw[thick,black,domain=2.2:3] plot[smooth] (\translateGEO+\x,{\phiOne(\x,0)},0);

    \draw[thick,black,domain=0.:3.] plot[smooth] (\translateGEO+\x,{\phiOne(\x,0)},2);

    \draw[thick,dashed,black,domain=2.3:3.] plot[smooth] (\translateGEO+\x,{\phiOne(\x,2)},0);
    \draw[thick,black,domain=0:2.2] plot[smooth] (\translateGEO+\x,{\phiOne(\x,2)},0);

    \draw[thick,black,domain=0.:3.] plot[smooth] (\translateGEO+\x,{\phiOne(\x,2)},2);

    \draw[cube hidden,black] (\translatePAR+2+\Pt,0,0) -- (\translatePAR+0+\Pt,0,0) -- (\translatePAR+0+\Pt,2,0);
    \draw[cube hidden,black] (\translatePAR+0+\Pt,0,0) -- (\translatePAR+0+\Pt,0,2);
    \draw[cube,black] (\translatePAR+0+\Pt,0,2) -- (\translatePAR+0+\Pt,2,2) -- (\translatePAR+2+\Pt,2,2) -- (\translatePAR+2+\Pt,0,2) -- cycle;
    \draw[cube,black] (\translatePAR+2+\Pt,2,2) -- (\translatePAR+2+\Pt,0,2) -- (\translatePAR+2+\Pt,0,0) -- (\translatePAR+2+\Pt,2,0) -- cycle;
    \draw[cube,black] (\translatePAR+2+\Pt,2,2) -- (\translatePAR+0+\Pt,2,2) -- (\translatePAR+0+\Pt,2,0) -- (\translatePAR+2+\Pt,2,0) -- cycle;

    \node (P) at (\translatePAR+1+\Pt,2.5) {P};

    \node (G) at (\translateGEO+1.5,2.5) {G};

}

\newcommand{\drawScalarField}[0]{
    \def\sfs{3}
    \draw[cube hidden,black] ({\translateRHO+\translatePAR+2+\Pt+\sfs},0,0) -- ({\translateRHO+\translatePAR+0+\Pt+\sfs},0,0) -- ({\translateRHO+\translatePAR+0+\Pt+\sfs},2,0);
    \draw[cube hidden,black] ({\translateRHO+\translatePAR+0+\Pt+\sfs},0,0) -- ({\translateRHO+\translatePAR+0+\Pt+\sfs},0,2);
    \draw[cube,black] ({\translateRHO+\translatePAR+0+\Pt+\sfs},0,2) -- ({\translateRHO+\translatePAR+0+\Pt+\sfs},2,2) -- ({\translateRHO+\translatePAR+2+\Pt+\sfs},2,2) -- ({\translateRHO+\translatePAR+2+\Pt+\sfs},0,2) -- cycle;
    \draw[cube,black] ({\translateRHO+\translatePAR+2+\Pt+\sfs},2,2) -- ({\translateRHO+\translatePAR+2+\Pt+\sfs},0,2) -- ({\translateRHO+\translatePAR+2+\Pt+\sfs},0,0) -- ({\translateRHO+\translatePAR+2+\Pt+\sfs},2,0) -- cycle;
    \draw[cube,black] ({\translateRHO+\translatePAR+2+\Pt+\sfs},2,2) -- ({\translateRHO+\translatePAR+0+\Pt+\sfs},2,2) -- ({\translateRHO+\translatePAR+0+\Pt+\sfs},2,0) -- ({\translateRHO+\translatePAR+2+\Pt+\sfs},2,0) -- cycle;

    \def\a{0.1}
    \def\b{0.77}
    \def\sfs{3}
    \path [shaded=1.5pt] ({\translateRHO+\translatePAR+\Pt+\shiftinvgamma+ \a-.13+\sfs},-0.5) .. controls (\translateRHO+\translatePAR+\Pt+\shiftinvgamma+\sfs-.3,0.4) .. ({\translateRHO+\translatePAR+\Pt+\shiftinvgamma+ \b-.13+\sfs},0.35);

}

\global\def\dataf(#1){{(1.75+#1/5*sin(360/10*#1))}}
\global\def\datat(#1,#2){{#2+2.4*exp(-10*#1*#1)}}
\global\def\smin{2}
\global\def\snull{8}

\global\def\stwo{6}

\global\def\sj{4}

\global\def\sn{3}
\global\def\dist{2.5}
\global\def\distNormal{8}

\newcommand{\drawTransLR}[0]{
    \pgfmathsetmacro{\z}{{\dataf(5.4)}}
    \foreach \x in {\stwo,\sj} {
        \pgfmathsetmacro{\y}{{\dataf(\x)}}
        \pgfmathsetmacro{\yi}{\y-\z}
        \pgfmathsetmacro{\xt}{\datat({\yi},\snull+\dist)}
        \draw[black,-triangle 45] ({\x},{\y}) -- (\xt,{\y});
    }
}
\newcommand{\drawFunc}[0]{
    \draw[black,->] ({\smin},0) -- (9,0) node[black,at end,below] {$s$};
    \draw[blue,domain=\sn:\snull] plot (\x,{\dataf(\x)});
}
\newcommand{\drawTrans}[0]{
    \drawFunc
    \draw[black,->] ({0.5+\smin},-0.5) -- ({0.5+\smin},2.5) node[black,at end,left] {$\rho_\gamma$};
    \draw[blue,domain=\sn:\snull] plot (\x,{\dataf(\x)});
    \draw[gray,-] (\sj,0) -- (\sj,{\dataf(\sj)}) node[black,at start,below] {$s_j$};
    \draw[black,->] ({-0.5+\snull+\dist},0) -- ({3+\snull+\dist},0) node[black,at end,below] {$c(r)$};
    \draw[black,->] ({0+\snull+\dist},-0.5) -- ({0+\snull+\dist},3) node[black,at end,left] {$r$};
}

\newcommand{\drawSubsampleLinear}[0]{

    \pgfmathsetmacro{\xL}{{\stwo}}
    \pgfmathsetmacro{\xR}{{\sj}}
    \pgfmathsetmacro{\yL}{{\dataf(\xL)}}
    \pgfmathsetmacro{\yR}{{\dataf(\xR)}}
    \draw[very thick,black,densely dashed] ({\xL},{\yL}) -- (\xR,{\yR});
}

\newcommand{\drawSubsample}[0]{

    \drawSubsampleLinear

    \pgfmathsetmacro{\z}{{\dataf(5.4)}}

    \pgfmathsetmacro{\x}{{\xL+0.75*(\xR-\xL)}}
    \pgfmathsetmacro{\y}{{\yL+0.75*(\yR-\yL)}}
    \pgfmathsetmacro{\yi}{\y-\z}
    \pgfmathsetmacro{\xt}{\datat({\yi},\snull+\dist)}

    \pgfmathsetmacro{\ytL}{{\yR}}
    \draw[red,fill,opacity=.5] (\xt,{\ytL}) -- (\xt,{\y}) -- ({\snull+\dist},{\y}) -- ({\snull+\dist},{\ytL}) -- cycle;
    \pgfmathsetmacro{\ytL}{{\y}}

    \pgfmathsetmacro{\x}{{\xL+0.5*(\xR-\xL)}}
    \pgfmathsetmacro{\y}{{\yL+0.5*(\yR-\yL)}}
    \pgfmathsetmacro{\yi}{\y-\z}
    \pgfmathsetmacro{\xt}{\datat({\yi},\snull+\dist)}

    \draw[red,fill,opacity=.5] (\xt,{\ytL}) -- (\xt,{\y}) -- ({\snull+\dist},{\y}) -- ({\snull+\dist},{\ytL}) -- cycle;
    \pgfmathsetmacro{\ytL}{{\y}}

    \pgfmathsetmacro{\x}{{\xL+0.25*(\xR-\xL)}}
    \pgfmathsetmacro{\y}{{\yL+0.25*(\yR-\yL)}}
    \pgfmathsetmacro{\yi}{\y-\z}
    \pgfmathsetmacro{\xt}{\datat({\yi},\snull+\dist)}

    \draw[red,fill,opacity=.5] (\xt,{\ytL}) -- (\xt,{\y}) -- ({\snull+\dist},{\y}) -- ({\snull+\dist},{\ytL}) -- cycle;
    \pgfmathsetmacro{\ytL}{{\y}}

    \pgfmathsetmacro{\x}{{\xL}}
    \pgfmathsetmacro{\y}{{\yL}}
    \pgfmathsetmacro{\yi}{\y-\z}
    \pgfmathsetmacro{\xt}{\datat({\yi},\snull+\dist)}
    \draw[gray,-] (\x,0) -- (\x,{\y)}) node[black,at start,below] {$s_{j-1}$};

    \draw[red,fill,opacity=.5] (\xt,{\ytL}) -- (\xt,{\y}) -- ({\snull+\dist},{\y}) -- ({\snull+\dist},{\ytL}) -- cycle;
    \pgfmathsetmacro{\ytL}{{\y}}

    \pgfmathsetmacro{\x}{{\xL+0.75*(\xR-\xL)}}
    \pgfmathsetmacro{\y}{{\yL+0.75*(\yR-\yL)}}
    \pgfmathsetmacro{\yi}{\y-\z}
    \pgfmathsetmacro{\xt}{\datat({\yi},\snull+\dist)}
    \draw[black,-triangle 45] ({\x},{\y}) -- ({\xt},{\y});
    \pgfmathsetmacro{\x}{{\xL+0.5*(\xR-\xL)}}
    \pgfmathsetmacro{\y}{{\yL+0.5*(\yR-\yL)}}
    \pgfmathsetmacro{\yi}{\y-\z}
    \pgfmathsetmacro{\xt}{\datat({\yi},\snull+\dist)}
    \draw[black,-triangle 45] ({\x},{\y}) -- ({\xt},{\y});
    \pgfmathsetmacro{\x}{{\xL+0.25*(\xR-\xL)}}
    \pgfmathsetmacro{\y}{{\yL+0.25*(\yR-\yL)}}
    \pgfmathsetmacro{\yi}{\y-\z}
    \pgfmathsetmacro{\xt}{\datat({\yi},\snull+\dist)}
    \draw[black,-triangle 45] ({\x},{\y}) -- ({\xt},{\y});

}

\newcommand{\drawNormal}[0]{

    \pgfmathsetmacro{\z}{{\dataf(5.4)}}

    \pgfmathsetmacro{\yL}{{\dataf(\sj)}}
    \pgfmathsetmacro{\yR}{{\dataf(\stwo)}}
    \pgfmathsetmacro{\yi}{\yR-\z}
    \pgfmathsetmacro{\xtR}{\datat({\yi},\snull+\distNormal)}
    \draw[red,fill,opacity=.5] (\xtR,{\yL}) -- (\xtR,{\yR}) -- ({\snull+\distNormal},{\yR}) -- ({\snull+\distNormal},{\yL}) -- cycle;

    \foreach \x in {\stwo,\sj} {
        \pgfmathsetmacro{\y}{{\dataf(\x)}}
        \pgfmathsetmacro{\yi}{\y-\z}
        \pgfmathsetmacro{\xt}{\datat({\yi},\snull+\distNormal)}
        \draw[black,densely dashed,-triangle 45] ({\x},{\y}) -- (\xt,{\y});
    }

    \draw[black,->] ({-0.5+\snull+\distNormal},0) -- ({3+\snull+\distNormal},0) node[black,at end,below] {$c(r)$} node[black,left=20pt,below=10pt] {\small normal};
    \draw[black,->] ({0+\snull+\distNormal},-0.5) -- ({0+\snull+\distNormal},3) node[black,at end,left] {$r$};

    \draw[black,->] ({-0.5+\snull+\dist},0) -- ({3+\snull+\dist},0) node[black,left=20pt,below=10pt] {\small supersampling};
}

\newcommand{\drawProposedSurf}[0]{

    \def\scale{1.2}

    \foreach \x in {0,1,2,3,4} {
        \draw[white] ({\x*\scale},{0}) -- ({\x*\scale},{4*\scale});
        \draw[white] ({0},{\x*\scale}) -- ({4*\scale},{\x*\scale});
    }

    \draw[very thick,black!30!green] plot [smooth cycle] coordinates {(1.2*\scale,2*\scale) (0.8*\scale,3.2*\scale) (2*\scale,2.8*\scale) (3.2*\scale,3.2*\scale) (2.8*\scale,2*\scale) (3.2*\scale,0.8*\scale) (2*\scale,1.2*\scale) (0.8*\scale,0.8*\scale) };
    \draw[fill=black!30!green,opacity=0.2] plot [smooth cycle] coordinates {(1.2*\scale,2*\scale) (0.8*\scale,3.2*\scale) (2*\scale,2.8*\scale) (3.2*\scale,3.2*\scale) (2.8*\scale,2*\scale) (3.2*\scale,0.8*\scale) (2*\scale,1.2*\scale) (0.8*\scale,0.8*\scale) };

    \draw[very thick,dashed, blue] (1.2*\scale,2*\scale) -- (0.8*\scale,3.2*\scale) -- (2*\scale,2.8*\scale) -- (3.2*\scale,3.2*\scale) -- (2.8*\scale,2*\scale) -- (3.2*\scale,0.8*\scale) -- (2*\scale,1.2*\scale) -- (0.8*\scale,0.8*\scale) -- (1.2*\scale,2*\scale);

}

\newcommand{\drawVoxSurf}[0]{

    \def\scale{1.2}
    \draw[very thick,black!30!green] plot [smooth cycle] coordinates {(1.2*\scale,2*\scale) (0.8*\scale,3.2*\scale) (2*\scale,2.8*\scale) (3.2*\scale,3.2*\scale) (2.8*\scale,2*\scale) (3.2*\scale,0.8*\scale) (2*\scale,1.2*\scale) (0.8*\scale,0.8*\scale) };
    \draw[fill=black!30!green,opacity=0.2] plot [smooth cycle] coordinates {(1.2*\scale,2*\scale) (0.8*\scale,3.2*\scale) (2*\scale,2.8*\scale) (3.2*\scale,3.2*\scale) (2.8*\scale,2*\scale) (3.2*\scale,0.8*\scale) (2*\scale,1.2*\scale) (0.8*\scale,0.8*\scale) };

    \foreach \x in {0,1,2,3,4} {
        \draw[densely dotted] ({\x*\scale},{0}) -- ({\x*\scale},{4*\scale});
        \draw[densely dotted] ({0},{\x*\scale}) -- ({4*\scale},{\x*\scale});
    }
    \def\siz{.1}

    \foreach \x in {0.5,1.5,2.5,3.5} {
        \foreach \y in {0.5,1.5,2.5,3.5} {
            \filldraw[fill = blue, draw = black] ({(\x-\siz)*\scale},{(\y-\siz)*\scale}) rectangle ({(\x+\siz)*\scale},{(\y+\siz)*\scale});
        }
    }
    \filldraw[ball color = red] ({2*\scale},{1.5*\scale}) circle (\radius);

}